\documentclass{article}
\usepackage[utf8]{inputenc}
\usepackage[numbers]{natbib}
\usepackage{amsmath,amssymb}
 
\usepackage[usenames,dvipsnames]{xcolor}

\usepackage{tcolorbox}
\tcbuselibrary{skins,breakable}
\tcbset{enhanced jigsaw}
\newcommand{\bQ}{\mathbf{Q}}
\definecolor{DarkRed}{rgb}{0.5,0.1,0.1}
\definecolor{DarkBlue}{rgb}{0.1,0.1,0.5}

\colorlet{YellowOrange}{RawSienna}
\newenvironment{proofof}[1]{\par{\noindent \bf Proof of #1:}}{\qed\par}

\usepackage{hyperref}
\hypersetup{
colorlinks=true,
pdfnewwindow=true,
citecolor=ForestGreen,
linkcolor=DarkRed,
filecolor=DarkRed,
urlcolor=DarkBlue
}
\usepackage{cleveref}

\usepackage{amsthm}
\newtheorem{lemma}{Lemma}
\newtheorem{claim}{Claim}

\newtheorem{definition}{Definition}

\newtheorem{observation}{Observation}

\newtheorem{theorem}{Theorem}
\usepackage{fullpage}

\usepackage{graphicx}

\usepackage{xcolor}
\newcommand{\esh}[1]{\textcolor{red}{[Eshwar: #1]}}

\newcommand{\hdn}{homogeneous noise model}

\usepackage{algorithmic}
\usepackage{algorithm}
\usepackage{comment}

\usepackage[shortlabels]{enumitem}

\newcommand{\fs}{\sqrt{\frac{\log n}{n}}}
\newcommand{\E}{\mathbb{E}}
\newcommand{\Mh}{\mathbf{Q}_h}

\usepackage[shortlabels]{enumitem}

\newcommand{\ignore}[1]{}

\newcommand{\ohp}{overwhelmingly high probability}
\newcommand{\hide}[1]{}

\usepackage[shortlabels]{enumitem}
\newcommand{\by}{\mathbf{y}}
\newcommand{\bX}{\mathbf{X}}
\newcommand{\NL}{\mathsf{NL}}
\newcommand{\lc}{\mathsf{lc}}
\newcommand{\rc}{\mathsf{rc}}
\newcommand{\bA}{\mathbf{A}}
\newcommand{\coloneqq}{:=}
\usepackage[shortlabels]{enumitem}

\newcommand{\p}{\widetilde{p}}
\newcommand{\h}{\widehat{h}}

\newcommand{\algfirst}{\textsf{Completion}}
\newcommand{\algsecond}{\textsf{Build-Subtree}}
\newcommand{\algoverall}{\textsf{Topology-Reconstruction}}
\newcommand{\algthird}{\textsf{Partition}}
\newcommand{\algweights}{\textsf{Tree-reconstruct-weight}}

\usepackage{caption}

\title{Reconstructing Ultrametric Trees from Noisy Experiments
}
\author{Eshwar Ram Arunachaleswaran \and Anindya De \and Sampath Kannan}

\begin{document}

\maketitle

\begin{abstract}
The problem of reconstructing evolutionary trees or phylogenies is of great interest in computational biology. A popular model for this problem assumes that we are given the set of leaves (current species) of an unknown binary tree and the results of `experiments' on triples of leaves $(a,b,c)$, which return the pair with the deepest least common ancestor. If the tree is assumed to be an \textit{ultrametric} (i.e., with all root-leaf paths of the same length), the experiment can be equivalently seen to return the closest pair of leaves. In this model, efficient algorithms are known for reconstructing the tree. 

In reality,  since the data on which these `experiments' are run is itself generated by the stochastic process of evolution, these experiments are noisy. In all reasonable models of evolution, if the branches leading to the three leaves in a triple separate from each other at common ancestors that are very close to each other  in the tree,  the result of the experiment should be close to uniformly random. Motivated by this, in the current paper, we consider a model where the noise on any triple is just dependent on the three pairwise distances (referred to as \emph{distance-based noise}). 
 Our results are the following: 
\begin{enumerate}
\item Suppose the length of every edge in the unknown tree is at least  $\tilde{O} (\frac{1}{\sqrt n})$ fraction of the length of a root-leaf path. Then, we give an efficient algorithm to reconstruct the topology of the unknown tree for a broad family of {distance-based noise} models. Further, we show that if the edges are asymptotically shorter, then topology reconstruction is information-theoretically impossible. 
\item Further, for a specific distance-based noise model -- which we refer to as the {\em{\hdn}} -- we show that the edge weights can also be approximately  reconstructed under the same quantitative lower bound on the edge lengths.
Note that in the noiseless case, such reconstruction of edge weights is impossible.  
\end{enumerate}



\end{abstract}

\section{Introduction}

The problem of clustering  is an important computational problem and a primitive that is used in multiple domains with the goal of grouping elements based on some underlying notion of distance in order to understand the relationship among them. 
In the standard clustering problem, 
the set of given elements is to be partitioned into
a few sets with the goal of putting similar elements 
in the same partition (captured by minimizing an objective function). 
 A natural and well studied variant (and generalization) of this problem is {{\em hierarchical clustering}}, where the goal is to find a hierarchical partition of the elements, in which  groups of elements form a nested structure. Equivalently, a hierarchical clustering can be thought of as a rooted tree with the elements at the leaves. Thus, the task of hierarchical clustering can be seen as the task of recovering the underlying unknown rooted tree.

  Naturally, canonical applications of the problem of hierarchical clustering are settings where there is an underlying tree structure -- examples include 
  learning evolutionary trees of a set of species and evolutionary trees of languages. In particular, the problem of reconstructing evolutionary trees or phylogenies from data about extant species is an important one in computational biology~\cite{saitou1987neighbor, semple_steel_2003, farach1999efficient, Mossel_07_distorted_phylogeny, Kearney_97_soda, ailon2005fitting, Chor_quartets, daskalakis2006optimal, erdHos1999few, kannan1996determining, zke2018}  and 
  is the principal motivation for this paper.

In order to define a hierarchical clustering problem, we need the precise  notion of similarity, as well as the mode by which the algorithm gets access to this information. The formulation that is closest to the one in this paper is from \cite{kannan1996determining}. Here, the evolutionary tree is assumed to be an \textit{ultrametric} binary tree, which is a weighted, rooted tree in which all root-leaf paths have the same length. This assumption is often justified in the computational biology literature based on the so-called molecular clock hypothesis, whereby the lengths of edges correspond to the evolutionary time that they represent. Then since all extant species are alive today and they are the leaves of this tree, they have all evolved for the same length of time. Consequently,  all root to leaf paths have the same length. 
The model in \cite{kannan1996determining} assumes that we are able to perform experiments on any 3 extant species (leaves) $a,b,$ and $c$ and the result of the experiment (alternately, referred to as query) is the pair that is closest together, i.e., has the most recent least common ancestor. In this model, the authors~\cite{kannan1996determining} give
efficient algorithms for reconstruction of the tree topology -- in fact, they give several procedures, each obtaining a different tradeoff between the running time and the number of experiments (i.e., queries). 


A principal shortcoming of \cite{kannan1996determining} is the assumption that the experiments do not have any noise -- on any triple $(a,b,c)$, the queries always returns the closest pair. However,  experiments are often noisy and thus it is natural to ask if we can design a tree-reconstruction algorithm which is tolerant to noise in the experiments. Several works~\cite{brown2013fast, Wy_phylogeny_2008, zke2018}
have explored this theme. In particular, in~\cite{zke2018}, the authors gave an algorithm to reconstruct the (topology of the) ultrametric tree with $O(n\log n)$ queries even in presence of noise.
 However, in all the previous works~~\cite{brown2013fast, Wy_phylogeny_2008, zke2018}, the noise  is identical across queries -- in other words, each experiment is assumed to give  an incorrect answer with some fixed probability $p<1/2$. 

 In this paper, we study the tree reconstruction problem under a broad family of noise models where the noise on any triple $(a,b,c)$ is just dependent on the three pairwise distances between
the leaves $a$, $b$ and $c$ -- note that by the ultrametric  assumption, the two largest distances are the same. We refer to such noise models as {\em distance-based noise models}. 
Our motivation is that if the data at the leaves is generated by a process like evolution then the data at each leaf is the result of a set of stochastic mutations encountered on the path from the root to that leaf. If we have three leaves $a,b,$ and $c$ where the least common ancestor of $a$ and $b$ is at distance 1 from the leaves, while the least common ancestor of $c$ with either $a$ or $b$ is at distance $1+ \epsilon$, then 
 the expected number of mutational differences between $a$ or $b$
on the one hand and $c$ on the other hand, is quite close to the expected number of mutational differences between $a$ and $b$. Any experiment trying to assess which pair is closest based on mutational differences would therefore have a good probability of identifying the wrong pair in this situation.
Finally, similar to~\cite{brown2013fast} (as well as many other results in the phylogenetic reconstruction literature), we assume that the noise in each experiment  is {\em permanent} -- i.e., repeating the same experiment always yields the same outcome. Since repetitions of an experiment will use the same noisy data, this assumption is justified. This naturally rules out repeating the same experiment as a way to {\em denoise the answers}, thus making the algorithm design more challenging. 

{\bf Our results:}  We now give an overview of our results. First of all, by rescaling the edge lengths (alternatively referred to as weights), we can assume that all root to leaf paths are of unit length. 
 For our first algorithmic result, we define a so-called {\em general noise model}. 
 This is any distance-based noise model that satisfies 
 two mild assumptions we refer to as the {\em monotonicity} and {\em anti-Lipschitzness} conditions. Roughly speaking, these conditions say that (i) the probability of getting the correct pair is always greater than  $1/3$ and (ii) Fixing the larger distance among pairs in the triple, the probability of the experiment returning the closest pair is sufficiently sensitive to the distance between the closest pair. The exact conditions are described  in Section~\ref{sec:model}.

 Our first algorithmic result shows that  in the general noise model, 
 as long as each edge length is at least $\tilde{\Omega}(n^{-1/2})$ (where $n$ is the number of leaves), there is an algorithm that takes the results of the $\binom{n}{3}$ experiments and reconstructs the topology of tree with high probability (Theorem~\ref{theorem:topology}).  We also show a matching lower bound -- namely, if the minimum edge length is $\tilde{o}(n^{-1/2})$, then it is information-theoretically impossible to recover the topology of the tree (Theorem~\ref{theorem:necessary}). Thus, together Theorem~\ref{theorem:topology} and Theorem~\ref{theorem:necessary} give the minimal requirements under which topology of the tree can be recovered in the general noise model.

For our second algorithmic result, we explore a special instance of the general noise model which we refer to as the {\em \hdn}. Let us denote this model by $\Mh(\cdot)$ -- then, on the triple $(a,b,c)$, the probability of returning the pair $(a,b)$ is given by 
\[
\Pr[\Mh(a,b,c) = (a,b)] = \frac{d(a,c) + d(b,c)}{2(d(a,b) + d(b,c) + d(a,c))}, 
\]
 where $d(\cdot, \cdot)$ denotes the distance function on the tree. 
 Under natural `boundary conditions' that the  probability of any pair being returned should approach 1/3 as the 3 pairwise distances approach each other, and that the probability of the closest pair being returned should approach some higher constant value when the other two distances tend to infinity,  $\Mh(\cdot)$ is essentially the only probability function that is a ratio of linear functions. One particularly appealing feature is that the model is invariant upon rescaling of the distance function $d(\cdot, \cdot)$. For the \hdn, we can achieve a significantly stronger result than Theorem~\ref{theorem:topology}. In particular, in Theorem~\ref{theorem:weights}, we show that as long as all the edge weights are at least $\tilde{\Omega}(n^{-1/2})$, there is an efficient algorithm to approximately reconstruct the edge weights. In other words, for the \hdn, we can not just recover the topology of the tree but the actual distances between leaves. Such a reconstruction of the edge weights is information-theoretically impossible in models such as
 ~\cite{kannan1996determining, brown2013fast, Wy_phylogeny_2008, zke2018} where either the queries have no noise or the noise only depends on the order between the distances and not their precise values. We remark that our techniques for  reconstructing distances are quite general, and should be applicable to other instantiations of the general noise model. Determining the general conditions under which the entire ultrametric can be reconstructed is left as a topic for future work.

\section{Related Work} \label{sec:related}

The problem of reconstructing evolutionary trees has received a lot of attention over the years. There are many formulations of this problem based on the type of data available and the objective function being optimized. Most formulations assume that the observed data is on extant species or leaves of an unknown tree. Objective functions seek to capture properties of the evolutionary process, with the hope that the optimal tree under an objective function is in fact the true evolutionary tree. The most popular formulations are distance-based methods~\cite{farachkannanwarnow,erdHos1999few, erdos1999few, saitou1987neighbor, mossel2017distance, daskalakis2013alignment} (where we are given a matrix of distances between leaves and we want to find the best-fitting edge-weighted tree), character-based methods~\cite{agarwalaf-b,kannanwarnow,mossel2005learning, MosselS07, steel2016phylogeny}, where we want to explain the evolution of different characters, each taking on a state in each extant species using the fewest number of state changes, and likelihood methods~\cite{neyman1971molecular, felsenstein1981evolutionary,farach1999efficient, roch2017phase}, where we assume that evolution is a stochastic process drawn from a family of processes, and want to estimate the most likely parameters. In all cases the data observed at the leaves is the result of the stochastic process of evolution. These formulations and other related types lend themselves naturally to the model considered in this paper. 

As noted earlier, the closest formulation to that in our paper is the one introduced by~\cite{kannan1996determining} on learning an ultrametric tree through experiments involving three leaves. 
This paper motivated several follow-ups~\cite{zke2018, brown2013fast, Wy_phylogeny_2008} with closely related models. In particular, in~\cite{zke2018}, the authors  considered a noisy version of these experiments, with each experiment independently failing 
with some probability $p < 1/2$. 
In the current paper, we consider the same problem but with a different, incomparable noise model. Similar to 
~\cite{zke2018}, the noise in each experiment is independent. However, unlike~\cite{zke2018}, the probability of getting an incorrect answer is not the same for every triple. In particular, for three leaves $a$, $b$ and $c$, where the pairwise distances are close to each other (in the ultrametric tree), the probability of getting the correct answer can be as small as $1/3 + \theta(\frac{\log n}{\sqrt{n}})$. An additional feature of our model is that each experiment can only be performed once (similar to \cite{brown2013fast}). 
In contrast, \cite{zke2018} allows for repetition of the same experiment multiple times with fresh randomness each time. 
Finally, we remark that while \cite{zke2018} allows for repetition of the same experiment, they view the number of experiments (equivalently the query complexity) as a key measure of performance of their algorithm -- in fact, their topology reconstruction algorithm has query complexity $O(n \log n)$ (which is essentially optimal). In contrast, the focus of this paper is to identify a broad class of noise models under which tree reconstruction is possible.

Besides evolutionary biology, ``distance based noise models" have also been studied in other reconstruction problem. In~\cite{tamuz2011adaptively}, Tamuz \emph{et~al.}~ study the following problem: there are $n$ elements with an unknown embedding in the Euclidean space. The algorithm gets noisy answers to {\em relative similarity queries} 
 and the goal is to reconstruct this embedding. More precisely, the algorithm can query any 
triple $(a,b,c)$ with the underlying semantics being ``Is $a$ closer to $b$ or to $c$?". On such a query, it gets the pair $(a,b)$ with probability $\frac{d(a,c)}{d(a,b) + d(a,c)}$  (and otherwise the pair $(a,c)$ is returned). Here $d(\cdot)$ is the underlying distance metric. We note that the model is both in form and spirit, very similar to the \hdn studied in Theorem~\ref{theorem:weights}. Indeed, as the distances  $d(a,b)$ and $d(a,c)$ approach each other, the response to the query $(a,b,c)$ is basically a coin flip. On the other hand, if one of the distances is much smaller than the other, then the probability of returning the closer pair approaches $1$. Along similar lines, Van der Maaten et al.~\cite{van2012stochastic} study the problem of learning a low dimensional embedding of a set of elements in Euclidean space based upon seeing the closest pair in a triplet. Just as in our model, each possible closest pair appears with a probability that depends upon an underlying dissimilarity/ distance function relating these elements.

Distance-based models are also quite popular in the ranking literature. In particular, in the well-known Bradley-Terry-Luce (BTL) model~\cite{Bradley-Terry, Luce}, 
there are $n$ elements where the $i^{th}$ element is assigned (an unknown) weight $w_i \in [0,1]$ -- thus defining a total order on these elements. 
The algorithm queries pairs $(i,j)$ and is returned $i$ with probability $w_i/(w_i+w_j)$. Note that this can be interpreted as a noisy comparison query where the probability of returning the larger element depends on the relative scores of the two elements. The goal in the BTL model is to recover the underlying ranking given these noisy comparison queries. Again, the BTL model bears strong syntactic resemblance to our \hdn (from Theorem~\ref{theorem:weights}). 
 In fact, similar to the current paper, in the 
BTL model, each query can be made at most once. We also  note that higher arity generalizations of the BTL model have also been explored in literature~\cite{plackett1975analysis, mcfadden1973conditional} (under the name multinomial logistic model).

Outside of distance based noise models, there is large body of literature in computer science which aims to model relations between elements by an (unknown) embedding in some metric space. The algorithm makes relational queries and gets noisy responses where the noise is governed by the hidden embedding. Several models  including the famous stochastic block model~\cite{mcsherry2001spectral, mossel2018proof, decelle2011asymptotic, abbe2017community} and its variants~\cite{chen2020near} fit this motif and the current paper can be seen as yet another instantiation of this general framework.  Further examples from the world of machine learning include~\cite{jamieson2011active},~\cite{kleindessner2014uniqueness},~\cite{agarwal2007generalized},~\cite{hoffer2015deep},~\cite{schultz2003learning}.


\label{section:notation}
\section{Model, Notation, and Preliminaries} \label{sec:model}
There is an underlying weighted tree $T$ with the weights constrained such that the distance function between leaves is an ultrametric. The set of leaves of the tree is represented by $L(T)$.  We assume that the height of the tree $h(T)$ is normalized to $1$.

{\bf Distance-Based Noise  Model:} For each triple of leaves $(a,b,c)$, we perform an experiment and get back one of the pairs $(a,b)$, $(b,c)$, or $(c,a)$ probabilistically. Such an experiment will be denoted by $Q(a,b,c)$. Repeating an experiment produces the same answer, and results of distinct experiments are independent of each other. 
Recall that in an ultrametric, the largest two of the three pairwise distances are equal.
Thus we will model the probabilities of different answers as a function of just two distances --- $d_1$, the distance
between the closest pair of leaves and $d_2$ the distance between either of the two other pairs in the triple. (So $d_1 \le d_2$. ) 
In our model, the two incorrect pairs have equal probabilities of being returned, which is justified because their pairwise distances are the same.
Thus, we define two probability functions: $p_{\textsc{correct}}$ and $p_{\textsc{incorrect}}$ where $p_{\textsc{correct}}$ denotes the probability that the closest pair is returned and $p_{\textsc{incorrect}}$ denotes the probability that  one of the other pairs is returned. Thus $\forall d_1, d_2, \ p_{\textsc{correct}}(d_1, d_2) + 2 p_{\textsc{incorrect}}(d_1, d_2)  = 1$. 

We impose some mild conditions on the probability functions. 
First, we naturally insist that $p_{\textsc{correct}} > p_{\textsc{incorrect}}$, since otherwise the output of the experiment is not useful. Second, we require that the probability of returning the correct pair is sufficiently sensitive to the change in distance. In mathematical terms, for any $d_2$ and $0 < d_1 < d_2$, $\frac{ \partial p_{\textsc{correct}}(d_1, d_2)}{\partial d_1} \le -\varepsilon$ 
for some constant $\varepsilon  > 0$. We are implicitly also assuming that the probability functions are continuous in each coordinate since we are assuming that their partial derivatives are defined. 


 We will refer to this model as the \textbf{distance-based noise model} or the \textbf{general noise model}. 



When it comes to reconstructing weights (Section~\ref{sec:weight}), we show that it is possible for a specific instantiation of the distance based model, called the \textbf{homogeneous model}, denoted by $\Mh(a,b,c) $. Recall the definition of this model from the introduction -  on the triple $(a,b,c)$, the probability of returning the pair $(a,b)$ is given by 
\[
\Pr[\Mh(a,b,c) = (a,b)] = \frac{d(a,c) + d(b,c)}{2(d(a,b) + d(b,c) + d(a,c))}, 
\]
where $d(\cdot, \cdot)$ denotes the distance function on the tree.


If $d(a,b)$ is the smallest of the 3 pairwise distances, then the probability of the experiment returning $(a,b)$ is between 1/3 and 1/2. Since  the other two distances, $d(a,c)$ and $d(b,c)$ are equal in an ultrametric,  the experiment has equal probability of returning $(a,c)$ or $(b,c)$.  
Thus for pairs $a, b$ that are very close and $c$ that is much farther, the probability of getting the result $(a,b)$ approaches .5, while for triples $(a,b,c)$ whose least common ancestors are very close, the probability of getting any pair approaches $1/3$. In the introduction we provided intuition on why this is a natural model.

{
Even in this simple model, we cannot hope to reconstruct arbitrary ultrametrics as the following example shows. Suppose the underlying tree is a balanced binary tree, where the edge at depth $i$ has weight $C \cdot 2^{2^i}$. Let the height of the tree be $\log n$ and the constant $C$ is chosen so that any root to leaf path has weight $C \cdot \sum_{i=0}^{h-1}2^{2^i} =1$. Now, consider any three leaves $a$, $b$ and $c$ such that the least common ancestor for any of the pairs is at height $h/2$. Further, for any three leaves $x,y,z$, let us call $\Mh(x,y,z)$ to be $\delta$-random if any of the pairs is returned with probability $1/3 \pm \delta$. Then, 
the following can be easily verified. 
(i) The experiment $\Mh(a,b,c)$ is $\exp(-\Theta(n))$-random. 
(ii) If $x \not =a,b,c$ is any other leaf, then the experiments, $\Mh(a,b,x)$, $\Mh(b,c,x)$ and $\Mh(a,c,x)$ are $\exp(-\Theta(n))$-random.  
(iii) If $x,y \not = a,b,c$ are two other leaves, then
the experiments, $\Mh(a,x,y)$, $\Mh(b,x,y)$ and $\Mh(c,x,y)$ are $\exp(-\Theta(n))$-random.   
From the above, it easily follows that using just $\binom{n}{3}$ experiments, the relative topology of the leaves $a$, $b$ and $c$ cannot be resolved. 
}

\ignore{
Even in this simple model, we cannot hope to reconstruct arbitrary ultrametrics as the following example shows: the underlying tree is a balanced binary tree, with each edge at depth $i$ having weight $2^i$. (These weights can be normalized to make each root-leaf path have length 1.)  Let $h$ be the height of the tree. Consider any 3 leaves both of whose least common ancestors are in the top $\log n / 2$ layers. The pair with the lowest least common ancestor has negligible advantage over the other two pairs. For any 3 such leaves, $a,b,$ and $c$, and for any other leaves $x$ and $y$,
it can be checked that any experiment involving two of $a,b,$ and $c$ and $x$  has indistinguishable probability of producing each of the 3 results, regardless of which pair out of $a,b,$ and $c$ is used in the experiment. Similarly any experiment involving  $x$ and $y$ and one of $a,b,$ and $c$ has indistinguishable probability of producing $(x,y)$ regardless of which of $a,b,$ or $c$ is used. Thus, with only $n^3$ experiments and negligible advantage from each, it will be impossible to learn even partial information about the topology in the top half of the tree. }

Thus, we will need to impose some conditions on the ultrametric to make the problem tractable.  Specifically, we show that a lower bound on the length of each edge is necessary and sufficient (up to log factors) for reconstructing the topology in the general model, and reconstructing the weights in the homogeneous model.

Without loss of generality, we will assume that the tree is a full binary tree, since an internal node with 1 child does not affect the response to any query and can be eliminated. 

Through this paper, when we refer to events having \ohp, we mean a probability of at least $1- \frac{1}{n^6}$. 
Since we will consider at most $o(n^5)$ such events, using the union bound, we can assume that all of them happen with very high probability (at least $1-\frac{1}{n}$), and condition our analysis upon this event.

We fix some standard notation for full binary trees that will be used in our algorithms.

{\bf Subtrees: } By  a \textit{subtree} of some tree $T$, we will mean the entire tree rooted at some internal node of $T$. (Thus we use the term ``subtree'' in a more restrictive manner than usual.) For any tree $T$, $L(T)$ denotes the set of leaves of $T$. We will also refer to the set of all leaves in the tree by $L$. 

{\bf Subtree-Induced Partition:} If $T_B$ is a subtree of $T$, it naturally partitions $L(T) - L(T_B)$ into buckets $S_1, S_2, \cdots S_k$ , where $x$ and $y$ are in the same bucket if and only if for any $z \in L(T_B)$, the least common ancestors of $x$ and $z$, and of $y$ and $z$ are the same. An alternative characterization of these buckets is that $x$ and $y$ are in the same bucket if and only if for any $z$ in $T_B$, the closest pair out of the triple $(x,y,z)$ is $(x,y)$.  Each bucket can be thought of as a subtree hanging off from the path from the root of $T_B$ to the root of $T$. Thus there is a natural order on the buckets that is defined by this path, with $S_1$ being the bucket closest to $T_B$ and $S_k$ the farthest. 
A visual depiction is shown in Figure~\ref{fig:total_order}.
For any $j \in [k]$, the set of leaves in $T_B, S_1, S_2 \cdots S_j$ form the leaves of a subtree.

 \begin{figure}[h]
	\centering
	\includegraphics[width = 6cm]{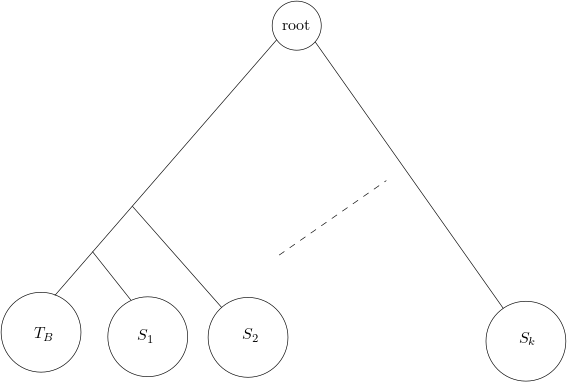}
	\caption{Partition of buckets with respect to Subtree $T_B$}	
	\label{fig:total_order}
\end{figure}

{\bf Induced Topology:}   For any subset of leaves $S$, the induced topology on $S$ is defined by removing all leaves outside of $S$ and removing internal nodes that now have only one child. (Since we are talking about weighted trees, when we have an internal node $v$ with one child, we replace the two edges incident on $v$ by a single edge whose weight is the sum of the weights of the two edges.)  It is not hard to see that the weighted tree obtained by this process will define an ultrametric on $S$. As a special case, when $T_B$ is a subtree of $T$, we will denote the induced topology on the leaves \textbf{not} in $T_B$ by $T - T_B$.
We will also define a slightly different induced topology where we replace $T_B$ by a single leaf (of $T_B$). We think of this operation as taking the quotient of $T$ with respect to $T_B$, and denote the resulting topology by $T\, / \, T_B$.

By the \textit{topology of a triple} of leaves $(x,y,z)$ we mean the induced subtree of $T$ with just these leaves. This topology is completely specified by specifying the pair among $x,y,$ and $z$ that has the least common ancestor of smallest height.

Finally, we use some standard concentration inequalities and  results about measures of statistical distance in our paper. These can be found in Appendix~\ref{sec:con} (Supplementary Material).

\section{Reconstructing Full Binary Trees}
\label{section:full_trees}

\label{sec:topology}

This section is devoted to the proof of the following result.

\begin{theorem}
\label{theorem:topology}
Let $T$ be a weighted full binary tree tree where all root-leaf paths are of length 1. Suppose that  the weight on each edge is at least $\tau \fs$, for some large constant $\tau$. Then, given access to the \textbf{general model} on the leaves of this tree, there exists an efficient algorithm $\algoverall$ that, with high probability, exactly reconstructs the topology of this tree.
\end{theorem}

We start by providing a high-level description of the algorithm $\algoverall$. We want to infer an unknown tree $T$ on a given set of leaves $L$. We are given the results of experiments on all triples $(a,b,c) \in L$. This result is one of the three possible pairs  with probabilities specified by the distance-based noise model (defined in Section~\ref{sec:model}). 
We use the phrase `resolving the topology' of a subtree $T'$ to mean that we know the rooted tree representation of $T'$.

Before describing the algorithm, we make an important observation regarding the probability of getting the correct answer.

\begin{observation}
\label{obs:advantage}
The assumption that all edge weights are at least $\tau \sqrt{(\log n /n )}$ implies that $d_1 + 2 \tau \sqrt{(\log n /n )} \le d_2$ for the distances $d_1, d_2$ involved in every experiment $Q(a,b,c)$. Thus, using the properties of the model, we observe that $p_{\textsc{correct}} \ge p_{\textsc{incorrect}} +2  \tau \sqrt{(\log n /n )}$.
\end{observation}

Our algorithm works by resolving the topologies of small subtrees, and then stitching these together until all of $T$ is resolved.

\begin{enumerate}

\item In a bottom-up manner by combining sibling subtrees, we build a ``base'' tree $T_B$ containing between $\sqrt{n}$ and $2 \sqrt{n}$ leaves that is a subtree of $T$ with high probability. (This is as large a tree as we can build to be confident that we have a subtree of $T$.)

\item We use the same idea to build a "pivot" tree $T_P$  on about the same number of  leaves outside of $L(T_B)$. With high probability, $T_P$ will also be a subtree of the induced tree on $L - L(T_B)$.

\item Using the fact that $|L(T_B) | * |L (T_P)| = \Omega (n)$,  we partition the leaves in $L(T) - L(T_B)$ into 3 parts (some possibly empty) --- leaves in buckets to the left of  $T_P$ (i.e., buckets with smaller indices than the buckets that the leaves of $T_P$ come from), leaves in the same bucket as $T_P$, and leaves in buckets to the right of $T_P$ (i.e., buckets with smaller indices than the buckets that the leaves of $T_P$ come from).

\item  We show that if a subtree excludes $\Omega (n)$ leaves then we can infer its entire topology with high probability. Likewise, if it contains $\Omega (n)$ leaves, we can infer the topology of the complement of the subtree with high probability. Using these two facts, we fully resolve the topology of all but one of the 3 parts in the previous step and recurse on the unresolved part. When this part has fewer than 
$\frac{11n}{12}$ leaves, we can infer its topology directly from the leaves outside, and the recursion bottoms out.

\end{enumerate}

Before describing the proofs of the different parts used in our algorithm, we state  a simple lemma about sums of independent random variables that recurs in multiple proofs. The proof of this lemma is in Appendix~\ref{app:topo}.

\begin{lemma}
	\label{lemma:enough_independent_sum}
	Let $L \ge \frac{n}{16}$, let $X_1, X_2, \cdots X_L$ and $Y_1, Y_2,\cdots Y_L$ be two sets of independent $0-1$ random variables such that for all $i \in [L]$, $\E[X_i] \le \E[Y_i] - c  \sqrt{(\log n /n )}$, for a sufficiently large constant $c$. Let $X = \sum_{i=1}^l X_i$ and let $Y = \sum_{i=1}^l Y_i$ Then, with high probability, i.e. at least $1- 1/n^6$, we have $Y >X + 24 \sqrt{n \log n}$. 
	
	Alternatively, if $X_i$ and $Y_i$ are identical random variables for each $i$, with high probability, we have $|X - Y| < 12 \sqrt{n \log n}$.
\end{lemma}

We will now elaborate on the algorithm. 
To begin, we state a lemma that establish the claims made in Step 4 above.


\begin{lemma}
\label{lemma:large_enough}
Let $T$ be a tree with $n$ leaves and $T'$ a subtree.
\begin{enumerate}
\item If $T'$ contains at least $\frac{n}{12}$ leaves, 
given the set $L(T')$, there exists an algorithm  $\algfirst$ that, with high probability, resolves the topology of the quotient $T\,/\,T'$. (Recall that this quotient is arrived at by collapsing $T'$ to a single leaf and taking the induced topology on the resulting set of leaves.)

\item If $T'$ has at most $\frac{11n}{12}$ leaves, given the set $L(T')$, the algorithm $\algfirst$ resolves  the induced topology on this set with high probability. 
\end{enumerate}

\end{lemma}

The proof of this lemma can be found in Appendix~\ref{app:topo}.

We now fill in the details of Step 1 in our algorithm outline above, describing and analyzing an algorithm that constructs an approximately $\sqrt{n}$ sized subtree within the induced tree of a  large enough set of leaves.

\begin{lemma}
\label{lemma:build_small_trees}
Given a subset $S \subseteq L$ of leaves with $|S| \ge \frac{n}{12}$,  denoting the (unknown) induced topology on $S$ by $T_S$, there exists an algorithm $\algsecond$ that, whp, finds a subtree $T'$ of  $T_S$ such that $|L(T')| \in [\sqrt{n}, 2 \sqrt{n}]$.
\end{lemma}


\begin{proof}
The algorithm $\algsecond$ builds subtrees of $T_S$ from the ground up, starting with each leaf in $S$ as a singleton subtree. Let $T_1, T_2, \cdots T_l$ represent subtrees of $T_S$ whose leaves form a disjoint partition of $S$ - in each step the algorithm, with high probability, combines two of these subtrees to form a larger subtree of $T_S$. Let $L_i$ represent the set of leaves of $T_i$. Using this procedure and starting with the initial configuration, the algorithm repeatedly keep applies this step until the size of the largest subtree exceeds $\sqrt{n} - 1$. Since each step can at most double the size of the largest subtree, we obtain a subtree of $T_S$ of size in the range $[\sqrt{n}, 2 \sqrt{n}]$.

Now, we describe the key step (combining subtrees) of the algorithm $\algsecond$.
For each pair $T_i, T_j$, a score $s_{ij}$ is generated in the following manner. Fix arbitrary $a \in L_i$, and $b \in L_j$ and define a $0-1$ random variable $X^x_{ij}$ for each leaf $x \in S \setminus (L_i \cup L_j)$ to be $1$ if the experiment $Q(a,b,x)$ gives $(a,b)$ as the answer, and 0 otherwise.
\[s_{ij} := \sum_{x \in S \setminus (L_i \cup L_j
)} X^x_{ij} \]
We claim that (whp) the tree pair with the highest score is in fact a ``sibling - tree pair'', i.e., a pair of subtrees of $T_S$ that can be combined through a shared parent internal node to form a subtree of $T_S$. There always exists such a sibling-tree pair, since these leaf-disjoint subtrees of $T_S$ themselves form a full binary tree.

We start with a simple observation: For any pair $T_k, T_l$ that is not a sibling pair, there is a sibling pair
$T_i, T_j$ such that their least common ancestor is strictly below the root of either $T_k$ or $T_l$.

To see this we consider the full binary tree resulting from starting with $T$ and collapsing each set of leaves $L_i$ into a single node and look at the subtree rooted at the least common ancestor of $T_k$ and $T_l$. There must be a deepest internal node in this tree whose children are sibling subtrees $T_i$ and $T_j$, completing the proof. Figure~\ref{fig:score_winner}, in Appendix~\ref{sec:figures} shows such a set of trees.


\begin{claim}
With high probability, $s_{kl} < s_{ij}$.
\end{claim}

\begin{proof}

Intuitively, this claim is true because the number of leaves that are used to generate the scores of both pairs $T_i, T_j$ and $T_k, T_l$ is sufficiently large  (at least $\frac{n}{16}$), and that these leaves all favor the pair $T_i, T_j$ (by at least $\theta( \sqrt{(\log n /n )})$ each). Further, the number of leaves that are used to generate only one of the scores is bounded by $2 \sqrt{n}$. Consequently, these leaves do not alter the signal created by comparing the action of the leaves used to generate both scores. 

We introduce some notation to formalize the above intuition. Let $S_1 : = S \setminus (L_k \cup L_l)$ and $S_2 : = S \setminus (L_i \cup L_j)$. $S_1$ is the set of leaves used to compute the score $s_{kl}$ and $S_2$ is the set of leaves used to compute the score $s_{ij}$.  We can write $S_1 = A \cup B_1$ and $S_2 = A \cup B_2$ where $A := S_1 \cap S_2$, $B_1 : =  L_i \cup L_j$ and $B_2 := L_k \cup L_l$. Since $|L_k|, |L_l|, |L_i|, |L_j| \le \sqrt{n}$, we can assume that $|A| \ge \frac{n}{16}$ and $|B_1|, |B_2| \le 2 \sqrt{n}$.

Using the condition on minimum edge length and the bound on the partial derivative of the function $p_{\textsc{correct}}(d_1,d_2)$, it is easy to see (Ref~Fig~\ref{fig:score_winner}) that :
\[ \forall x \in A: \E[X^x_{ij}]  - \E[X^x_{kl}] \ge \frac{\tau \sqrt{\log n}}{9 \sqrt{n}}  \]



 Expanding $S_1$ and $S_2$, we get: 
\[s_{kl} = \sum_{x \in A}X^x_{kl} + \sum_{y \in B_1}X^y_{kl}\]

\[s_{ij} = \sum_{x \in A}X^x_{ij} + \sum_{x \in B_2}X^x_{ij}\]

Taking the difference and using linearity of expectation, we get:
\begin{align*}
\E[s_{ij}] - \E[s_{kl}] &= \sum_{x \in A} (\E[X^x_{ij}]  - \E[X^x_{kl}]) + \sum_{y \in B_2} E[X^y_{ij}] -  \sum_{y \in B_1} E[X^y_{kl}] \\
&\ge \sum_{x \in A} (\E[X^x_{ij}]  - \E[X^x_{kl}]) -  \sum_{y \in B_1} E[X^y_{kl}] \\
&\ge \frac{\tau \sqrt{n \log n}}{144} - 2 \sqrt{n}
\end{align*}

where the last inequality derives from the fact that each $X^y_{kl}$ is a $0-1$ random variable and hence has expectation upper bounded by $1$.

Since both $s_{kl}$ and $s_{ij}$ are each sums of at least $n/16$ independent $0-1$ random variables, we can use the Chernoff bound to conclude that their expectations do not differ from their value by more than $24 \sqrt{n \log n}$ with very high probability. Putting together these inequalities, we conclude that, for large enough $\tau$, we have $s_{kl} < s_{ij}$ with high probability. The proof of this claim assumes that $O(n^2)$ high probability events occur simultaneously, since it suffices that a correct subtree pair beats every wrong subtree pair, of which there are at most $n^2$.

\end{proof}

An immediate corollary of this claim, via an application of the union bound is that the highest scoring pair is in fact a sibling-tree pair. This gives an efficient algorithm that correctly combines subtrees of $T_S$ to form a larger subtree of $T_S$. The proof of correctness of one call to the algorithm $\algsecond$ assumes the simultaneous occurrence of $O(n^3)$ high probability events - since the subroutine to combine subtrees is used
at most $n$ times (from the fact that each such subroutine introduces an internal vertex, and there are most $n$ of them).
\end{proof}

We show how to partition the leaves within a contiguous interval of buckets onto either side of a $\sqrt{n}$ sized subtree in the induced topology on the leaves within this interval.

\begin{lemma}
\label{lemma:partial_order_refinement}
Let $T_B$ be a subtree of $T$ with $|T_B| \ge \sqrt{n}$ and let $I$ be a contiguous interval in the partition of buckets with respect to $T_B$. Let $T_P$ be a subtree of the tree induced on $L(I)$ with between $\sqrt{n}$ and $2\sqrt{n}$ leaves. Then there exists an algorithm $\algthird$ that partitions  $L(I) - L(T_P)$ into 3 sets $P_1, P_2,$ and $P_3$ such that $P_1$ consists of all leaves in lower indexed buckets than the leaves of $T_P$, $P_2$ consists of leaves in the same bucket as $T_P$, and $P_3$ consists of leaves in higher numbered buckets. In the case that $L(T_P)$ comprises leaves from more than one bucket $P_1$ and $P_2$ are empty.

\end{lemma}


\begin{proof}
When $P_2$ is non empty, all leaves in $T_P$ are equidistant from $T_B$, implying that $P_2$ must be contained within a single bucket $S_i$.

We describe the algorithm $\algthird$ that does the partitioning. For each leaf $x \in L - (T_B \cup T_P)$, introduce $0-1$ random variables $X^x_{ab}$ and $Y^x_{ab}$ for each $a \in T_B,  b\in T_P$. Each such leaf $x$ has $2n$ such associated random variables, since $|T_B|, |T_P| \ge \sqrt{n}$. $X^x_{ab}$ is set to $1$ if $Q(a,b,x)$ is $(a,x)$ and $0$ otherwise. $Y^x_{ab}$ is set to $1$ if $Q(a,b,x)$ is $(a,b)$ and $0$ otherwise. Let $X^x = \sum_{a \in T_B, b \in T_P} X^x_{ab}$ and $Y^x = \sum_{a \in T_B, b \in T_P} Y^x_{ab}$. If $X^x - Y^x > 24 \sqrt{n \log n}$, then, $x$ is placed in set $P'_1$. If If $Y^x - X^x > 24 \sqrt{n \log n}$, then, $x$ is placed in set $P'_3$. Otherwise if $|Y^x - X^x| \le 24 \sqrt{n \log n}$, $x$ is placed in $P'_2$. We claim that for $i \in [3], P'_i = P_i$ with high probability.


The first case is $x \in P_1$, we know that $d(a,b) \ge d(a,x) + 2 \tau \fs$, for all $a \in T_B, b \in T_P$, by the definition of $S_1$ and the minimum weight condition. Using Observation~\ref{obs:advantage}, we conclude that $\E[X^x_{ab}] - \E[Y^x_{ab}] \ge \frac{\tau \sqrt{\log n}}{3 \sqrt{n}}$.  Using Lemma~\ref{lemma:enough_independent_sum}, we conclude that the $X^x - Y^x > 24 \sqrt{n \log n}$, with high probability.


Next, consider $x \in P_3$, we know that $d(a,x) \ge d(a,b) + 2 \tau \fs$, for all $a \in T_B, b \in T_P$, by the definition of $S_1$ and the minimum weight condition (Observation~\ref{obs:advantage}). Through similar analysis as the previous case, we conclude $x$ is correctly placed in $P'_3$ whp.

Finally, we see the case of $x \in P_2$. We know that $d(a,x) = d(a,b) > d(b,x)$  for all $a \in T_B, b \in T_P$, by the definition of $S_2$.Thus, $X^x_{ab}$ and $Y^x_{ab}$ are identical random variables, since the incorrect answers to any query appear with equal probability. 
Using the second part of Lemma~\ref{lemma:enough_independent_sum}, we get the desired result about the comparison of $X^x$ and $Y^x$, with high probability.
Thus, $x$ is correctly placed in $P'_2$ whp.

In total, we need $O(n)$ high probability events to happen simultaneously (one for partitioning each leaf) in this proof of correctness.

We index the buckets of the contiguous interval $I$ as $S_{j_1}, S_{j_1 +1}, \cdots$. For the final part of the claim, we analyze the case where $T_P$ consists of leaves from more than one bucket.  Let us index these buckets as $S_{i_1}, S_{i_1+1}, \cdots$ where $i_1 \ge j_1$. Clearly, $T_P$ must consist of all leaves in an interval of buckets, since leaves within a bucket are always closer to each other than leaves in another bucket (recall that the leaves of each bucket are the leaves of a subtree of $T$, Ref~Fig~\ref{fig:total_order}). Due to this property, $P_2$ must be empty. Recall that the set of leaves in $T_B, S_1, S_2 \cdots S_i$ forms a tree for any index $i$. Consequently, every leaf $l \in L(S_i)$ is closer to any leaf in $T_B, S_1, S_2, \cdots S_{i-1}$ than to a leaf in $S_j$ with $j > i$. Now, consider the case in which $T_P$ consists of leaves contained in a union of buckets - $S_{i_1} ,S_{i_1 + 1}  \cdots S_{i_2}$ where $i_1 > j_1$. The leaves of $T_P$ cannot be a set of leaves of a tree in the induced topology of $L(I)$, since any leaf in $S_{i_1}$ is closer to a leaf in $S_{j_1}$ (which is part of the set $L(I)$) than to a leaf in $S_{i_2}$. We have a contradiction, invalidating the assumption $i_1 > j_1$. Thus, we must have $i_1 =j_1 $, implying that $P_1$ is also empty. 

\end{proof}

We now have all the details in place for a full description of the algorithm $\algoverall$. The algorithm creates a base tree $T_B$ using $\algsecond$ (Lemma~\ref{lemma:build_small_trees}) at the beginning.
We use the same technique to build a subtree $T_P$  of $L(T) - L(T_B)$, also  with $\Theta (\sqrt{n})$ leaves.

For any two leaves $x, y$  not in $L(T_B) \cup L(T_P)$, we will say that $x \preceq_{T_B} y$ iff  $x$'s bucket has index less than or equal to $y$'s.
$x$ and $y$ are in the same bucket iff $x \preceq_{T_B} y$ and $y \preceq_{T_B} x$.  (When the base tree is clear from the context, we will drop the subscript on
the relation symbol.) Using $T_B$ and $T_P$ and the technique in Lemma~\ref{lemma:partial_order_refinement}, we determine the $\preceq$ relationship for all pairs
$x,y \in L - (L(T_B) \cup L(T_P))$. Based on this, we partition these leaves into $P_1, P_2,$ and $P_3$ using $\algthird$( as described in Lemma~\ref{lemma:partial_order_refinement}).

Since $T_P$ is a subtree of the tree induced on $L - L(T_B)$, there are two possibilities for what $T_P$ looks like with regard to the partition into buckets based on $T_B$.
Either $T_P$ is entirely contained within one of these buckets, or it is the union of an initial interval of buckets.  In the latter case, $P_1$ and $P_2$ are empty.

If the number of leaves in any of the 3 parts is less than $\frac{11n}{12}$ we will say that that part is small. 

\begin{claim}
\label{claim:finish_the_job}
The topology of a small part can be resolved completely with high probability using the algorithm $\algfirst$ (from Lemma~\ref{lemma:large_enough}).  
\end{claim}

\begin{proof} 
If $P_1$ is small, noting that $T_B \cup P_1$ form a subtree and appealing to Lemma~\ref{lemma:large_enough}, the topology of $P_1$ is inferred using $\algfirst$. If $P_2$ is small, there are two cases: if $P_1$ is large, an appeal to Lemma~\ref{lemma:large_enough} suffices, noting that $P_2$ lies outside of a large subtree. Else, $P_3$ must be large, and once again, we appeal to Lemma~\ref{lemma:large_enough}, but this time using the fact that $P_2$ is part of a small subtree. Finally, when $P_3$ is small, it is very easy to see how the same Lemma implies that the topology of $P_3$ can be resolved.
\end{proof}

Since there are at least two small parts in the partition, at most one part will be unresolved at the end of this process. Let us call this part $R$. 
In the case that $T_P$ did not lie entirely within one bucket,
$R$ will be $P_3$ (the only non-empty part). 

We now explain how our algorithm recursively resolves the topology of this contiguous interval.
If the pivot tree $T_P$ in the current step lies within $R$ (we must be in the case where $T_P$ lies entirely in one bucket, and the unresolved part is this bucket), then we recurse on $R \, / T_P$,  the quotient of $R$ with respect to $T_P$. Since $T_P$ has at most $2 \sqrt{n}$ leaves, the number of nodes in the quotient will still be large.
This recursive call finds a new base tree, a partition with respect to this tree, and so on. If, on the other hand, $T_P$ lies outside of the unresolved part, we recurse on the unresolved part, keeping the same base tree and finding a pivot tree within this part. We will maintain the invariant that $R$ is a contiguous interval of buckets with respect to the current base tree.  If $|R|$  is less than $11n/12$, we can reconstruct it using Claim~\ref{claim:finish_the_job}.

In the case that we change the base, $R$ is an entire bucket, and hence a subtree of the overall tree. Again if $R$ becomes small enough, we can reconstruct it directly by appealing to Lemma~\ref{lemma:large_enough}.

Thus, we will be able to complete the reconstruction of the whole tree. One final subtle point that the reader may wonder about: Could we lose enough nodes from
the quotient-finding operations that $R$ becomes small while the number of nodes outside of $R$ is also small? A careful accounting shows that this does not happen: Each time we compute a quotient and remove at most $2 \sqrt{n}$ leaves, we have also created a new base tree with at least $\sqrt{n}$ leaves, which lies outside and to the left of $R$. Thus if $R$ becomes small due to the quotient operation, the number of leaves on the left side must be large enough to allow reconstruction of $R$.


To conclude the proof of Theorem~\ref{theorem:topology}, we argue that the algorithm $\algoverall$ succeeds with high probability. The key observation is that we use the various subroutines $\algfirst$, $\algsecond$ and $\algthird$ at most $O(\sqrt{n})$ times each. This is because we reduce the size of the unresolved part by $\sqrt{n}$ using only a constant number of calls to these subroutines. Each of them assume at most $O(n^3)$ high probability events to be simultaneously true to succeed, implying that we only need a total of $O(n^{7/2})$ high probability events to all be true for the algorithm to correctly recover the topology of the tree. Since each of them occurs with probability at least $1- \frac{1}{n^6}$, an application of the union bound gives us the desired result. Further, since each subroutine runs in polynomial time, the overall algorithm is also efficient.

\section{Weight Reconstruction}
\label{sec:weight}
In the previous section, we gave an algorithm to reconstruct the topology of the tree in the general model. In this section, we will show how to approximately reconstruct the edge weights in the homogeneous model. The precise theorem is stated below (Theorem~\ref{theorem:weights}). Assuming that each root to leaf path has unit weight, our algorithm can reconstruct the tree as long as each edge has weight at least $\tau \cdot \log n / \sqrt{n}$. Note that the condition required here is stronger than Theorem~\ref{theorem:topology} where it suffices that each edge weight is $\Omega(\sqrt{\log n / n})$. 
\begin{theorem}~\label{theorem:weights}
Let $T$ be a weighted full binary tree tree such that the weights induce an ultrametric on the distances between leaves and that the weight on each edge is at least $\tau  \frac{\log n}{\sqrt{n}}$, for some large constant $\tau$. Then there exists an algorithm $\algweights$, which given access to the \textbf{homogeneous model} on the leaves of this tree,  reconstructs the weight of each edge with high probability within an additive error $\kappa \frac{\log n}{\sqrt{n}}$, where $\tau \gg \kappa$).
\end{theorem}

 We explain the high-level strategy for procedure {\sf Tree-reconstruct-weight}. Instead of estimating the weight of each edge, we will give a procedure which for any vertex $v$, will estimate the weight of the path from root to $v$ up to an additive $\kappa \log n/ 2\sqrt{n}$. This trivially implies reconstruction of the weight of each edge up to $ \pm \kappa \log n/ \sqrt{n}$. Our algorithm assumes that $\tau \gg \kappa$. 

Since the homogeneous model is a special case of the general model (with $\epsilon = 1/6$) and the edge weights satisfy the condition required by Theorem~\ref{theorem:topology} on the edge weights, we can first run $\algoverall$ to reconstruct the topology of the tree (with high probability). The rest of the proof assumes we have access to the topology of the tree. The main workhorse for weight reconstruction is the idea that if we have two leaves $a$ and $b$ that lie in the left subtree of some node $v$, and the right subtree of $v$ has $\Omega (n)$ leaves, then we can  get a good approximation to the height of the least common ancestor of $a$ and $b$. A related idea is that if we have an internal node $v$ such that the product of the number of leaves in its left and right subtrees is $\Omega (n)$, then we can also get a good approximation to the height of $v$.
This leaves  the case of nodes $v$ for which neither of these properties is true and much of the technical difficulty of our algorithm is in handling such nodes (Lemma \ref{lemma:tricky_weights}). For such a node we get several coarse approximations of its height, which need to be combined carefully because of subtle dependence between these approximations.

We start by fixing some notation. Let $\NL(v)$ denote the number of 
leaves in the sub-tree rooted at vertex $v$. We will adopt the convention that the children of each node are ordered so that for any node $v$ with  right child $\rc$ and left child $\lc$, $\NL(\rc) \ge \NL(\lc)$. Further,  without loss of generality, we can also assume that the binary tree is full -- otherwise, if an internal node $v$ has exactly one child, we can collapse the two edges incident on $v$ into one. This modification does not affect the output probabilities associated with any triple of leaves. Finally, for any vertex $v$, we define $h_v$ as the ``height" of $v$ -- i.e., the total weight of any path from $v$ to a leaf in its subtree. Because the distance on the tree is an ultrametric, the choice of the leaf is immaterial. 

\begin{definition}
	A vertex $v$ is said to be {\em heavy} if $\NL(v) \ge \alpha n +1$ where $\alpha = \frac{1}{6}$. Otherwise, the  vertex is said to be {\em light}. By definition, the root $r$ is a heavy vertex.
\end{definition}

Next, given a rooted tree as above, we identify a distinguished vertex $v_f$ as follows. Let $P$ be the rightmost path from the root to a leaf . Starting from the root $r$, let the vertices in $P$ (in order) be $r$, $v_1 , \ldots, v_{\ell^\ast}$ where $v_{\ell^\ast}$ is the rightmost leaf. We define $f$ to be maximum index such that $v_f$ is a heavy vertex. 


Consider the root $r$, with left child $\lc_r$ and right child $\rc_r$. Observe that by our convention, $\rc_r$ is a heavy vertex. We now give a procedure to (approximately) compute $h_v$ for any vertex $v$ in the subtree rooted at $\lc_r$. 

\begin{lemma}~\label{lem:compute-vertex-lv}
	There is a procedure \textsf{Compute-light-tree} which 
	with high probability, computes $h_v$ for all vertices $v$ in the tree rooted at $\lc_r$ with accuracy $\Theta (\sqrt{\log n/n})$. 
\end{lemma}
\begin{proof}
	Consider any non-leaf vertex $v$ in the subtree of $\lc_r$. Since $v$ is a non-leaf vertex,  there is a leaf $a$ in its left subtree and $b$ in the right subtree. Additionally, since $\rc_r$ is a heavy vertex, there are at least $\alpha n$ leaves under it - indexed by $\{c_i\}_{i=1}^k$ where $k \ge \alpha n$ . Let $\bX_i$ be the indicator random variable that query $Q(a,b,c_i)$ returns $(a,b)$. Each $\bX_i$ is an independent (and in fact identically distributed) random variable such that 
	\[
	\Pr[\bX_i =1] = \frac{d(a,c_i) +d(b,c_i)}{2(d(a,b)+d(a,c_i) +d(b,c_i))}=\frac{4}{2(4+2h_v)} =\frac{1}{2+h_v}.
	\]
	Let us define $\bA  \coloneqq (1/k) \cdot (\sum_{i=1}^k \bX_i)$. Define $h'_v$ to be such that $\bA = \frac{1}{2+h'_v}$. Now, by Hoeffding's inequality (Theorem~\ref{theorem:hoeffding}), we get that with \ohp, 
	\[ \bigg|\bA - \frac{1}{2 + h_v} \bigg|\le \kappa_2 \sqrt{\frac{\log n}{n}}. \] 
	Thus
	\[
	\bigg|\frac{1}{2 + h'_v} - \frac{1}{2 + h_v} \bigg|\le \kappa_2 \sqrt{\frac{\log n}{n}}.
	\]
	This means that $$|h_v-h'_v| \le \kappa_2 \sqrt{\frac{\log n}{n}} (2 + h_v)(2 + h'_v).$$ 
	Note that $h_v \le 1$. Because  $\big|\frac{1}{2 + h'_v} - \frac{1}{2 + h_v} \big|\le 1/6$, it follows that $h'_v \le 4$. Thus, it follows that $|h_v-h'_v| \le \Theta (\sqrt{\log n/n})$. 
\end{proof}

Recall that we had define $P$ as the rightmost path in the tree with vertices $\{v_i\}_{i=1}^{\ell^\ast}$ such that $v_f$ is the {\em last} heavy vertex in the sequence. Next, we have the following lemma. 
\begin{lemma}~\label{lem:reconstruct-weight-rightpath}
	There is a procedure \textsf{Reconstruct-right-path} which with high probability computes $h_v$ up to $\pm \Theta (\sqrt{\log n /n})$ for any $v_\ell$ in the path $P$ where $\ell \le f$. 
\end{lemma}
\begin{proof}
	Consider vertex $v_\ell$, since $\ell \le f$, we know that $\NL(v_\ell) \ge \alpha n +1$. Let the number of leaves in the left (resp. right) subtree of $v_\ell$ be $s_\ell$ (resp. $s_r$). We have $s_\ell + s_r \ge \alpha n +1$. It immediately follows that the number of pairs of the form $(a_i, b_i)$ where $a_i$ is in the left subtree and $b_i$ is in the right subtree is at least $\alpha n$. Now, let $c$ be a leaf in the left subtree of the root $r$ -- we know that there exists at least one. 
	
	Let us now define $\bX_i$ to be the indicator variable that the query on triple $(a_i,b_i,c)$ returns $(a_i,b_i)$. Observe that each $\bX_i$ is an i.i.d.~Bernoulli random variable such that $\Pr[\bX_i =1] = \frac{1}{2+ h_{v_\ell}}$. Now repeating the same calculation as done at
	end of Lemma~\ref{lem:compute-vertex-lv}, we obtain an estimate $h'_v$ such that with \ohp, $|h'_v - h_v| = \Theta(\sqrt{\log n/n})$.


\end{proof}







\begin{lemma}
	\label{lemma:anchoring}
	Let $v$ be an internal vertex of the tree and let $v_\ell$ be a heavy vertex on path $P$ such that (i) $v_\ell$ is an ancestor of vertex $v$; (ii) There are $k \ge 100$ leaves in the sub-tree of $v_\ell$ that does not contain $v$. Let $a$ and $b$ be leaves in distinct subtrees of $v$ and $c$ a  vertex in the subtree of $v_l$ not containing $v$. Let $p$ be the probability that a query on the triple $(a,b,c)$ returns $(a,b)$. Note that $p$ is independent of the particular choices of $a,b,$ and $c$ as constrained above. Then, there is an algorithm that outputs an estimate $\hat{p}$ such that 
	\begin{enumerate}
		\item $|\hat{p} - p| = O\big(\sqrt{\log k /k}\big)$. \item $\hat{p}$ is an unbiased estimator of $p$. 
	\end{enumerate}
	
	We say that vertex $v$ has been anchored to vertex $v_l$ to obtain this estimate.
\end{lemma}
\begin{proof}
	
	We index the leaves in the sub-tree of $v_\ell$ not containing $v$ as $c_1,c_2,..c_k$. Let $c \in \{c_1, \ldots, c_k\}$. As noted in the statement of the lemma,
	the probability that $Q(a,b,c)$ returns $(a,b)$ is a constant $p$ independent of $c$. 
	
	The pairwise distances are $d(a,b) = 2h_v$ and $d(a,c) = d(b,c) = 2h_{v_\ell}$.  This implies that $p = \frac{h_{v_\ell}}{2h_{v_\ell} + h_v}$. Let $\bX_i$ be the indicator random variable that $Q(a,b,c_i)$ returns $(a,b)$. Note that 
	\[
	\Pr[\bX_i=1] = \frac{h_{v_\ell}}{2h_{v_\ell} + h_v}. 
	\]
	Now, define $\hat{p} \coloneqq (\sum_{i=1}^k \bX_i)/k$.
	Hoeffding's inequality (Theorem~\ref{theorem:hoeffding}) and linearity of expectation now immediately imply the claim. 
\end{proof}

\begin{lemma}~\label{lem:reconstruct-height-left-subtree}
	Suppose $v$ is an internal vertex which lies in the left subtree of $v_\ell$ -- where $v_\ell$ lies on the path $P$ and $1 \le \ell \le f$. 
	There is a procedure \textsf{Reconstruct-height-left-heavy} which given as input such a vertex $v$, reconstructs the height of $v$ up to error $O(\alpha^{-1} \cdot \sqrt{\log n /n})$. 
\end{lemma}
\begin{proof}
	We are given that $v_\ell$ is an ancestor of $v$. The right subtree of $v_\ell$  has at least $\alpha n$ leaves. Thus, applying Lemma~\ref{lemma:anchoring}, we can obtain an estimate $\hat{p}$ such that 
	\[
	\bigg| \hat{p} - \frac{h_{v_\ell}}{h_v + 2 h_{v_\ell}} \bigg| = O\bigg( \frac{1}{\alpha} \cdot \sqrt{\frac{\log n}{n}}\bigg). 
	\]
	Now, $h_v$ can be expressed in terms of $h_{v_\ell}$ and $p$ as follows: 
	\[
	h_{v} = \frac{h_{v_\ell} \cdot (1-2p)}{p} \coloneqq f(v_\ell,p). 
	\]
	Using Lemma~\ref{lem:reconstruct-weight-rightpath}, we can obtain an estimate $\widehat{h_{v_\ell}}$ such that 
	\[
	\big| h_{v_\ell}- \widehat{h_{v_\ell}} \big| = \Theta \bigg(\sqrt{\frac{\log n}{n}} \bigg). 
	\]
	The estimate of the procedure is $\widehat{h_v}$ defined as $f(\widehat{h_{v_\ell}}, \hat{p})$. We now observe that by triangle inequality, 
	\begin{eqnarray}~\label{eq:bound-hat-estimate-1}
		\big| f(\widehat{h_{v_\ell}}, \hat{p}) - f(h_{v_\ell},p) \big| \le \big| f({h_{v_\ell}}, \hat{p}) - f(h_{v_\ell},p) \big|  + \big| f(\widehat{h_{v_\ell}}, \hat{p}) - f({h_{v_\ell}}, \hat{p}) \big|. 
	\end{eqnarray}
	To bound the terms on the right hand side, observe that $h_{v_\ell} \le 1$ and $p \in [1/3,1]$. Thus,  by Lemma~\ref{lemma:anchoring}, we can assume that 
	$\hat{p} \ge 1/4$. 
	Consequently, 
	\begin{equation}~\label{eq:bound-hat-estimate-2}
		\big| f({h_{v_\ell}}, \hat{p}) - f(h_{v_\ell},p) \big| = O \bigg( \frac{|p-\hat{p}|}{p^2}\bigg) = O(|p-\hat{p}|). 
	\end{equation}
	\begin{equation}~\label{eq:bound-hat-estimate-3}
		\big| f({h_{v_\ell}}, \hat{p}) -  f(\widehat{h_{v_\ell}}, \hat{p})\big| = O\bigg( \frac{\big|h_{v_\ell} - \widehat{h_{v_\ell}} \big|}{\hat{p}^2}  \bigg) =  O\big( \big|h_{v_\ell} - \widehat{h_{v_\ell}} \big| \big). 
	\end{equation}
	Plugging \eqref{eq:bound-hat-estimate-2} and \eqref{eq:bound-hat-estimate-3} into \eqref{eq:bound-hat-estimate-1} gives us the stated claim. 
\end{proof}

Using Lemma~\ref{lem:reconstruct-weight-rightpath} and 
Lemma~\ref{lem:reconstruct-height-left-subtree}, we have constructed $h_v$ approximately (to additive $O(\pm \sqrt{\log n/n})$) for all vertices except the ones in the subtree of $v_{f+1}$. For such vertices $v$, we next show how to compute $h_v$ up to an additive $O(\log n/\sqrt{n})$. Note that this estimate is slightly worse than the ones achieved in Lemma~\ref{lem:reconstruct-height-left-subtree} and Lemma~\ref{lem:reconstruct-weight-rightpath}. 

The above lemmas 
reconstruct $h_v$ for any vertex in the left subtree of any heavy vertex except the last, as well as any vertex in the ``rightmost path". We next describe how to compute $h_v$ for the remaining vertices, i.e., the vertices in the subtree rooted at $v_{f+1}$. At this point, we also note that if
there was a vertex $v$ with two heavy children, then the obvious adaptation of Lemma~\ref{lem:compute-vertex-lv}
gives us the weight of every edge in the subtree rooted at $v$, and hence the entire tree $T$. However, in general, there is no guarantee that this will happen and we have to go for a more significantly more complicated procedure. Towards this, we prove the following crucial lemma.

\begin{lemma}
	\label{lemma:tricky_weights}
	There is a procedure \textsf{Reconstruct-Internal-left} which given any internal vertex $v$ in the subtree of $v_{f+1}$, outputs an estimate of $h_v$
	that has additive error  $ O(\log n /\sqrt{n})$  . 
\end{lemma}

This lemma completes the proof of Theorem~\ref{theorem:weights}, since we have successfully reconstructed the heights of all internal vertices of the tree within additive error $O(\frac{\log n}{\sqrt{n}})$. Since this proof is considerably more involved than the previous proofs, we present it in Appendix~\ref{app:weights}.

\section{Necessary Conditions}
\label{sec:necessary}

The goal of this section is to show that to  reconstruct the topology of the tree, it is necessary for each edge to have weight $\Omega(1/{\sqrt{n}})$ -- thus, essentially matching the lower bound assumption in Theorem~\ref{theorem:topology}. 
Recall that we normalize the edge weights so that the height (weighted root to leaf distance) of the tree is $1$.

In the theorem below, we give a nearly matching lower bound on the minimum weight of each edge even for topology reconstruction in our model. We also note that such edge weight lower bounds are commonplace in the literature on phylogenetic reconstruction\cite{felsenstein1981evolutionary,farach1999efficient,erdHos1999few}.

\begin{theorem}
\label{theorem:necessary}
Let $T$ be the set of weighted full binary trees tree such that the weights induce an ultrametric on the distances between leaves (within each tree).  Then, given access to the \textbf{homogeneous model}, as described in Section~\ref{sec:model}), on the leaves of this tree, there exists no algorithm that can reconstruct the topology of the tree if the edge weights can be as small as $ \frac{\rho}{\sqrt{n}}$ where $\rho$ is a sufficiently small constant ($\rho \le \frac{1}{100}$). {A fortiori,} we obtain the same lower bound for the general model as well. 
\end{theorem}


To prove this result, we construct two trees with
the following properties: 
\begin{enumerate}
\item They have distinct topologies and the weight of each edge is at least $\rho/\sqrt{n}$. 
\item It is information-theoretically impossible to distinguish between the trees with probability more than $0.51$ (in the homogeneous model). 
\end{enumerate}

The two trees $T_1$ and $T_2$ are as follows. Both have roots with identical weighted left subtrees $B$. The right subtrees of both $T_1$ and $T_2$ both have three leaves $a, b,$ and $c$ but with different induced topologies.
In $T_1$, $a$ and $b$ are sibling leaves with parent $p$. The parent of $c$ and $p$ is the node $q$, which is the right child of the root. In $T_2$, $a$ and $c$ are siblings, whose parent is $x$. $x$ and $b$ have a parent $y$, which is the right child of the root. All `corresponding' edge lengths are identical and in particular, the edges $(p,q)$ and $(x,y)$ have weight $ \frac{\rho}{\sqrt{n}}$.  The edges to the sibling pair of leaves in the right subtree of both trees have the same weight,  say $\frac{1}{3}$. The two trees we construct are shown in Figures~\ref{fig:Lower Bound Instance Tree 1} and~\ref{fig:Lower Bound Instance Tree $T_2$} in Appendix~\ref{sec:figures} (supplementary material).

\ignore{
The two trees we construct are shown in Figures~\ref{fig:Lower Bound Instance Tree 1} and~\ref{fig:Lower Bound Instance Tree $T_2$}.  The two trees are identical except for the relative topology of the leaves $a$, $b$ and $c$.  The left subtree $B_k$ is any full binary tree with weights satisfying the ultrametric property, and such that the length of the path from the root of the entire tree to any leaf in $B_k$ is $1$. The edge $(q,p)$ has weight $ \frac{\rho}{\sqrt{n}}$ and the edge $(a,p)$ has constant weight, say $\frac{1}{3}$, while the edge from the root to $q$ is weighted such that the root to leaf distance when the leaf is in the right subtree is also $1$.
}

\ignore{
\begin{figure}
\centering
\begin{minipage}{0.5\textwidth}
  \centering
  \includegraphics[width=.8\linewidth]{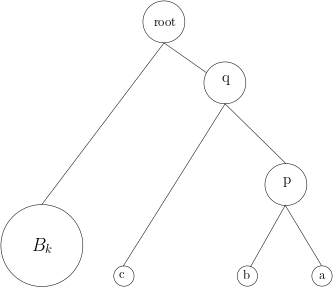} 
  \captionof{figure}{Lower Bound Instance Tree $T_1$}
  \label{fig:Lower Bound Instance Tree 1}
\end{minipage}%
\begin{minipage}{.5\textwidth}
  \centering
  \includegraphics[width=.8\linewidth]{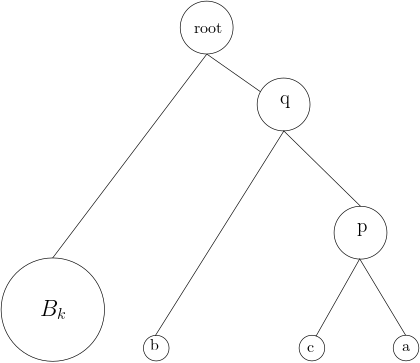}
  \captionof{figure}{Lower Bound Instance Tree 2}
    \label{fig:Lower Bound Instance Tree $T_2$}
\end{minipage}
\end{figure}
}

The rest of the proof of Theorem~\ref{theorem:necessary} has been moved to Appendix~\ref{app:lower}, in interest of space.


\newcommand{\etalchar}[1]{$^{#1}$}

\appendix

 \section{Proofs from Section~\ref{sec:topology} (Topology Reconstruction)}
\label{app:topo}

\subsection{Proof of Lemma~\ref{lemma:enough_independent_sum}}

	Since $X$ and $Y$ are each sums of at least $\frac{n}{16}$ independent $0,1$ random variables, we use the Hoeffding bound (Theorem~\ref{theorem:hoeffding}) to argue that the events $|X-\E[X]| \le 6 \sqrt{n\log n}$  and $|Y-\E[Y]| \le 6 \sqrt{n\log n}$ each happen with high probability, i.e with $1- \frac{1}{2n^6}$. Using the union bound, both events happen with probability at least $1- \frac{1}{n^6}$.
	
	For the first part of the lemma, since for each $i \in [L]$, $\E[X_i] \le \E[Y_i] - c \sqrt{\frac{\log n}{n}}$, using linearity of expectation gives : $\E[X] \le \E[Y] - c \sqrt{\frac{\log n}{n}}$. Using the triangle inequality, we conclude that $|X - Y| < 12 \sqrt{n \log n}$ is true with probability at least $1-  \frac{1}{n^6}$.
	
	For the second part of the lemma, since $X_i$ and $Y_i$ are identical random variables, $\E[X] = \E[Y]$.  Using the triangle inequality, we conclude that $Y >X + 24 \sqrt{n \log n}$ is true with probability at least $1-  \frac{1}{n^6}$.

\subsection{Proof of Lemma~\ref{lemma:large_enough}}
\label{appendix:goldilocks}
{\bf Part 1:} First, we describe how  $\algfirst$ reconstructs the buckets $S_1, S_2 , \cdots S_k$ that partition $L(T) - L(T')$. Let $L'$ be the set of leaves in $T'$. 


The following test is used to resolve the relative order of two leaves $x,y \in L \setminus L'$. For each $a \in L'$, the random variable $X_a$ is set to 1 if the result of the
experiment $Q(a,x,y)$ is $(a,x)$, and to 0 otherwise. Similarly, the random variable $Y_a$ is set to 1 if the result of this experiment is $(a,y)$, and to 0 otherwise.
 Let $X = \sum_{a \in L'} X_{a}$ and $Y = \sum_{a \in L'} Y_{a}$.

If $X - Y > 24\sqrt{n \log n}$, then we  declare that $x$'s bucket has a lower index than $y$'s. Otherwise, if $|X - Y| \le 24\sqrt{n \log n}$, we say that $x$ and $y$ are in the same bucket. We prove that this algorithm is correct with high probability. 

We first show that the above method of comparison works correctly for any pair of leaves $x,y$ in different buckets. Without loss of generality, $x \in S_i$ and $y \in S_j$ with $i<j$.  
Using Observation~\ref{obs:advantage} about experiment $Q(a,x,y)$, we conclude that $\E[X_a] - \E[Y_a] \ge 2 \tau \fs$. Using Lemma~\ref{lemma:enough_independent_sum}, we get the desired result about the comparison of $X$ and $Y$, with high probability.



We next extend the result to pairs of leaves $x,y$ in the same bucket $S_i$. We know that $d(a,x) = d(a,y) > d(x,y)$  for all $a \in L'$.Thus, $X_{a}$ and $Y_{a}$ are identically distributed random variables (since our noise model guarantees that the two incorrect answers to a experiment appear with equal probability). Using the second part of Lemma~\ref{lemma:enough_independent_sum}, we get the desired result about the comparison of $X$ and $Y$, with high probability.



Finally, $\algfirst$ labels the buckets generated by doing all pairwise tests using a standard topological sorting algorithm to recover the original labels. $O(n^2)$ high probability events are assumed to simultaneously occur as part of this proof, since we are comparing all pairs of leaves.

In the next phase, the algorithm $\algfirst$ resolves the closest pair of three leaves from the same bucket. A score $s_{ab}$ is associated with each pair of leaves $a,b \in S_i$ and a $0-1$ random variables $X^x_{ab}$ is defined for each $x \in L'$. Each pair $a,\,b$ has at least $\frac{n}{12}$ associated random variables. $X^x_{ab}$ is set to $1$ if $Q(a,b,x)$ is $(a,b)$ and $0$ otherwise. Let $s_{ab} = \sum_{x \in L'} X^x_{ab}$. Given three leaves $a,b,c \in S_i$  the pair with the highest score is declared to be the closest pair.  

We argue that this test succeeds with high probability. Let $(a,b)$ be the closest pair among $(a,b,c)$. We will show that with high probability $s_{ab} > s_{bc}$. (An identical argument helps establish that $s_{ab} > s_{ac}$.)   We observe that $d(b,c) \ge d(a,b) + 2 \tau \fs$ and $ d(a,b) < d(b,c) <  d(a,x) = d(b,x) = d(c,x)$ for all $x \in L'$. By the properties of the noise model (the bound on the partial derivative), we get $\E[X^x_{ab}] - \E[X^x_{ac}] \ge 2 \tau \fs$.  
Using Lemma~\ref{lemma:enough_independent_sum}, we get the desired result about the comparison of $s_{ab}$ and $s_{bc}$, with high probability.

 $O(n^3)$ high probability events are assumed to simultaneously occur as part of this proof, since we are separately arguing correct recovery of the closest pair for every possible set of three leaves.



{\bf Part 2:} Let $L'$ be the set of leaves in $T'$.  We describe how the algorithm $\algfirst$ recovers the topology of $T'$.

 A score $s_{ab}$ is associated with each pair of leaves $a,b \in L'$ and a  $0-1$ random variables $X^x_{ab}$ is defined for each $x \in L \setminus L'$. Each leaf pair $a,b$ has at least $\frac{n}{12}$ associated random variables. $X^x_{ab}$ is set to $1$ if $Q(a,b,x)$ is $(a,b)$ and $0$ otherwise. Let $s_{ab} = \sum_{x \in L \setminus L'} X^x_{ab}$. Given three leaves $a,b,c \in L'$, the closest pair is chosen as the pair with the largest score. The argument that this test succeeds with high probability uses Lemma~\ref{lemma:enough_independent_sum} in a manner similar to the proof of Part 1 and hence we omit it.  $O(n^3)$ high probability events are assumed to simultaneously occur as part of this proof, since we are separately arguing correct recovery of the closest pair for every possible set of three leaves.




\section{Proof of Lemma~\ref{lemma:tricky_weights} (Weights Reconstruction)} 
\label{app:weights}


Recall that the vertices in path $P$ are $v_0, \ldots, v_f$ (in order) where $v_0$ is the root and $v_f$ is the last heavy vertex. By definition, this  implies that the total number of leaves in the left subtrees of $v_0$, $v_1$, $\ldots$, $v_f$ is at least $(1-\alpha) n -1$. For $0 \le i \le f$, let $\mathsf{C}_i$ be the set 
of leaves in the left subtree of $v_i$. Then $\sum_{i=0}^f |\mathsf{C}_i| \ge (1-\alpha)n -1 > n/2$ (as $\alpha \le 0.49$). 
Let $\mathcal{A} \coloneqq \{v_i: | \mathsf{C}_i| < n/(4(1+f))\}$. It easily follows that $\sum_{i \not \in A} |\mathsf{C}_i| > n/4$. 

Let the elements of $\overline{\mathcal{A}}$ (in order) be $t_1, \ldots, t_k$. Let $c \in \mathsf{C}_{t_i}$ and let $a$ and $b$ be vertices in different subtrees of $v$. Then, note that the probability with which $Q(a,b,c)$ returns $(a,b)$ is $p_i$ where $p_i = h_{t_i}/(2h_{t_i} + h_v)$. We now apply Lemma~\ref{lemma:anchoring} and using $v_{t_i}$ as an anchor for $v$ --  with~\ohp ($1-\frac{1}{n^6}$), this gives us an estimate $\hat{p}_i$ such that 
\[
|\hat{p}_i - p_i| \le 4 \bigg( \sqrt{\frac{\log |\mathsf{C}_{t_i}|}{|\mathsf{C}_{t_i}|}}\bigg)
\]


At this point, one might ask whether any of the estimates $\hat{p}_i$ is good enough to construct a good estimate $\h_v$ for the height of vertex $v$. However, note that the best guarantee that we can give for any $|C_{t_i}|$ is at most $\theta(\sqrt{n} \log n)$, since the number of vertices on a root to leaf path (to which $f$ may be comparable) can potentially be as bad as $\frac{\sqrt{n}}{\log n}$. To demonstrate such an instance, consider the tree where each edge the rightmost path from the root to leaf path has weight $\theta(\frac{\sqrt{n}}{\log n})$, and each heavy vertex has $\theta(\sqrt{n} \log n)$ vertices in its left subtree. With such a guarantee, using similar techniques as in the proof of Lemma~\ref{lem:reconstruct-height-left-subtree}, we can only obtain an estimate for $h_v$ such that the additive error is upper bounded by $O(\frac{1}{n^{1/4}\log n})$. Such an estimate falls well short of our target of $O(\frac{\log n}{n})$ additive error.

To get a better estimate, we use the fact that the set of random variables $\{\hat{p}_i\}_{i \in [k]}$ are independent, owing to the fact that they are functions of disjoint sets of queries. We expect to see that errors in these random variables balance out when we aggregate them in some fashion. A natural approach is to use each $\hat{p}_i$ to construct an estimator $\h^i_{v} = \h_{t_i} (\frac{1}{\hat{p}_i} - 2)$ for $h_v$ and take the weighted average (weighted by the $|C_{t_i}|$s). However, this approach runs into difficulties, triggered by the fact that the random variables $\h^i_{v}$ are not unbiased estimators for $h_v$. To avoid this issue, we aggregate the $\hat{p}_i$'s directly and then recover a single estimate $\h_v$ for $h_v$ from the aggregated quantity. In particular, we focus on estimating the quantity $Q := \frac{\sum_i |C_{t_i}| p_i}{\sum_i |C_{t_i}|}$ through an estimator $\widehat{Q}$ and recovering a good estimate $\h_v$ for$h_v$ using the estimate $\widehat{Q}$. The proof is complete if we prove the following two claims.

\begin{claim}
	\label{claim:goodq}
	From the estimators $\{\hat{p}_i\}_{i \in [k]}$, there exists an algorithm to recover an estimator $\widehat{Q}$ for $Q = \frac{\sum_i |C_{t_i}| p_i}{\sum_i |C_{t_i}|}$ such that $|\widehat{Q} - Q| \le \theta \bigg (\frac{\log n}{\sqrt{n}} \bigg)$.
\end{claim}

\begin{claim}
	\label{claim:final_estimate_constructor}
	Given a good estimate $\widehat{Q}$ for $Q = \frac{\sum_i |C_{t_i}| p_i}{\sum_i |C_{t_i}|}$, i..e, with additive error $O \bigg (\frac{\log n}{\sqrt{n}} \bigg)$, there exists an algorithm $\textsf{Final- Estimate}$ that uses estimators $\{\h_{t_i}\}_{i \in [k]}$ to construct an estimator $\h_v^f$ for $h_v$ with additive error $O \bigg (\frac{\log n}{\sqrt{n}} \bigg)$.
\end{claim}

\begin{proofof}{Claim~\ref{claim:goodq}}
	Assume, for sake of a thought experiment, that each $\hat{p}_i$ is restricted to a range of $\theta \bigg( \sqrt{\frac{\log |\mathsf{C}_{t_i}|}{|\mathsf{C}_{t_i}|}}\bigg)$ around $p_i$. Then, if we take the weighted average $\frac{\sum_i |C_{t_i}| \hat{p}_i}{\sum_i |C_{t_i}|}$ of the $\hat{p}_i$s, the size of the range of each term $\frac{|C_{t_i}| \hat{p}_i}{\sum_i |C_{t_i}|}$ is at most $\theta \bigg(\frac{\sqrt{ |C_{t_i}|\log n}}{n} \bigg)$. Thus, the variance of the sum is upper bounded by $O(\frac{\log n}{n})$ (this uses the fact that $n \ge \sum_i |C_{t_i}| \ge \frac{n}{4}$). Using an exponential tail bound would lead us to the desired result under this thought experiment.
	
	Consequently, the natural approach is to directly take the weighted average  of the $\hat{p}_i$s, since this would be an unbiased estimator for $Q$. However, we only have the guarantee that each $\hat{p}_i$ is at most $\theta \bigg( \sqrt{\frac{|\mathsf{C}_{t_i}| \log n}{|\mathsf{C}_{t_i}|}}\bigg)$ away from $p_i$ as a high probability event rather than as absolute truth. To get around this roadblock, for each $i \in [k]$, we define a real valued random variable $\p_i$, and introduce the following coupling between $\hat{p}_i$ and $\p_i$ based on truncating $\hat{p}_i$:
	
	\[ \p_i = \begin{cases}
		\hat{p}_i, \text{ when } |\hat{p}_i - p_i| \le 4 \bigg( \sqrt{\frac{\log |\mathsf{C}_{t_i}|}{|\mathsf{C}_{t_i}|}}\bigg) \\
		p_i + 4 \bigg( \sqrt{\frac{\log |\mathsf{C}_{t_i}|}{|\mathsf{C}_{t_i}|}}\bigg), \text{ when } \hat{p}_i - p_i > 4 \bigg( \sqrt{\frac{\log |\mathsf{C}_{t_i}|}{|\mathsf{C}_{t_i}|}}\bigg) \\
		p_i -  4 \bigg( \sqrt{\frac{\log |\mathsf{C}_{t_i}|}{|\mathsf{C}_{t_i}|}}\bigg), \text{ otherwise }
	\end{cases}
	\]
	
	We complete the description of the algorithm by defining our final estimator : $\widehat{Q} := \frac{ |C_{t_i}| \hat{p}_i}{\sum_i |C_{t_i}|}$ .
	
	We observe that the random variables $\{\p_i\}_{i \in [k]}$ are independent. This follows from the fact that the random variables $\{\hat{p}_i\}_{i \in [k]}$ are themselves independent. As a consequence of this doing this truncation, the resultant $\p_i$ is no longer an unbiased estimator of $p_i$, however we show that this does not functionally affect us. To do so, we prove a claim showing that the expectations of the coupled random variables are very close to each other. 
	
	\begin{claim}
		\label{claim:close_means}
		For each $i \in [k]$, $ |\E[\p_i] - \E[\hat{p}_i]| \le \frac{2}{n^6}$.
	\end{claim}
	
	\begin{proof}
		We already know that $\Pr[|\hat{p}_i - p_i| \ge \kappa_1 \bigg( \sqrt{\frac{\log |\mathsf{C}_{t_i}|}{|\mathsf{C}_{t_i}|}}\bigg) ] \le \frac{1}{n^6}$. Additionally, the random variable $\hat{p}_i$ takes its value in the range $[0,1]$ since it is an empirical average. This also implies that $\E[\hat{p}_i] = p_i \in [0,1]$. Thus, we have:
		
		\[ |\E[\p_i] - \E[\hat{p}_i] | \le |\frac{1- p_i}{n^6} |  + |\frac{p_i}{n^6} | \]
		which gives us the desired result.
	\end{proof}

	Now, define $j_i := |C_{t_i}|$ for each $i \in [k]$. Define the function $g(\p_1,\p_2,..\p_k) := \frac{\sum_{i=1}^k j_i \p_i }{\sum_{i=1}^k j_i}$  on the domain $\mathbb{R}^k$ 
	Each random variable $\frac{ j_i \p_i }{\sum_{l=1}^k j_l}$ is within the interval $[a_i,b_i]$ such that $L_i = |b_i -a_i| = \theta ( \frac{\sqrt{j_i \log n}}{n})$.
	We know that $\sum_{i=1}^k L^2_i = \sum_{i=1}^k \theta (\frac{ j_i \log n}{n^2}) = \theta (\frac{\log n}{n})$.


	Using the above property, we appeal to the generalized Hoeffding bound (Theorem~\ref{theorem:generalized_hoeffding}).
	With probability at most $\frac{1}{n^6}$, we have $|g(\p_1,\p_2,..,\p_k) -  \E[g(\p_1,\p_2,..,\p_k)]| \ge \theta \bigg(\frac{\log n}{\sqrt{n}} \bigg)$. Using Claim~\ref{claim:close_means}, we know that $|\E[g(\p_1,\p_2,..,\p_k)] -  \E[\frac{ \sum_i j_i \hat{p}_i}{ \sum_i j_i}]| \le \frac{2}{n^6}$. Using linearity of expectation, we know that $\E[\frac{ \sum_i j_i \hat{p}_i}{ \sum_i j_i}]= \frac{ \sum_i j_i p_i}{ \sum_i j_i} = Q$.  Using the triangle inequality, the event $|\widehat{Q} - Q| \le \theta \bigg(\frac{\log n}{\sqrt{n}} \bigg)$ happens with high probability.
	

\end{proofof}

\begin{proofof}{Claim~\ref{claim:final_estimate_constructor}}
	
	First, we describe how to recover estimator $\h_v^f$ for $h_v$ from $\widehat{Q}$ using the estimates $\{\h_{t_i}\}_{i \in [k]}$.
	
	Define $j_i := |C_{t_i}|$ for each $i \in [k]$. We define a multivariate function $F$ that operates on inputs $a \in \mathbb{R}$ and $b = (b_1, b_2 ,\cdots b_k) \in \mathbb{R}^k$. 
	
	\[ F(a,b) :=\frac{ \sum_{i=1}^k \frac{j_ib_i}{2b_i + a}}{\sum_i j_i} \]
	
	
	Alternatively, we write $F(a,b) = \frac{\sum_{i=1}^k j_i f(a,b_i) }{\sum_i j_i}$ where $f(a,b_i) = \frac{b_i}{2b_i + a}$, i.e., we write $f$ as a convex combination of $k$ functions. For ease of notation, we call $\alpha_i = \frac{j_i}{\sum_{a=1}^k j_a}$. Note that all $\alpha > 0$ and $\sum_i \alpha_i = 1$.
	
	Let $H = \{ h_{t_1}, h_{t_2}, \cdots h_{t_k}\}$ and $\widehat{H} = \{ \h_{t_1}, \h_{t_2}, \cdots \h_{t_k}\}$. Observe that $Q = F(h_v,H)$. Intuitively, we want to invert the value $\widehat{Q}$ of the function $F$ using our estimates $H'$ to obtain an estimate $\h_v$ that is close to $h_v$. We prove a claim below that tells us how to do so.
	
	
	\begin{claim}
		There exists a unique real number $\h_v$ such that $F(\h_v, \widehat{H}) = \widehat{Q}$. Additionally, this real number can be found using numerical methods from $\widehat{Q}$ and $\widehat{H}$.
	\end{claim}
	
	\begin{proof}
		Fix the second argument of $F$ to be $\widehat{H}$, $F$ is now a univariate function on the first argument. This function is monotone decreasing since each univariate function $f(a,\h_{t_i})$ is monotone decreasing and $F(a,\widehat{H})$ is a convex combination of these functions. Thus, $F$ is invertible. Further, simple numerical methods such as binary search can be used to find $\h_v$ upto arbitrary precision such that $F(\h_v, \widehat{H}) = \widehat{Q}$.  
	\end{proof}
	
	The algorithm then makes the following final correction to the estimate : $\h_v^f \leftarrow \min \{\h_v, \min_i \h_{t_i} \} $ - there is a compelling reason to make such a correction - to ensure that the vertex $v$ always has estimated height that is no larger than the estimated height of one of its ancestors. Note that since we eventually show that the error in estimating edge weight is only a fraction of the minimum edge weight, such an event never happens in practice -  however, we include this correction for the sake of the analysis leading to the proof of that result.
	
	First, we show a simplified analysis, for the case that $\h_v \le \min_i \h_{t_i}$. Here the final estimator is $\h_v^f = \h_v$. We will later show a more involved analysis for when our estimate does not satisfy these conditions.

	For the sake of contradiction, assume that $|\h_v - h_v| \ge \Delta \frac{\log n}{\sqrt{n}}$, where $\Delta$ is a sufficiently large constant. Since  $|\widehat{Q} - Q| \le \theta (\frac{\log n}{\sqrt{n}} )$, we get:
	
	\begin{align*}
		|\sum_i \alpha_i \{f(h_v,h_{t_i}) - f(\h_v,\h_{t_i}) \}| &\le \theta (\frac{\log n}{\sqrt{n}} ) \\
		\implies |\sum_i \alpha_i \{(f(h_v,h_{t_i}) - f(\h_v,h_{t_i})) + (f(\h_v,h_{t_i}) - f(\h_v,\h_{t_i})) \}|  &\le \theta (\frac{\log n}{\sqrt{n}} )    \tag{*}
	\end{align*}
	
	Our approach is to show that each $A_i := |f(h_v,h_{t_i}) - f(\h_v,h_{t_i})|$ is sufficiently larger (by at least some $\Omega (\frac{\log n}{\sqrt{n}} ) $  than any $B_i := |f(\h_v,h_{t_i}) - f(\h_v,\h_{t_i})|$. Using the triangle inequality would finish the proof. To this purpose, recall that $|\h_{t_i} - h_{t_i}| \le \theta (\fs)$.
	
	Intuitively, we wish to show that the function $f(a,b_i)$ changes rapidly with change in the first input $a$ (large partial derivative in our interval of interest) while it is relatively more stable with change in the second input $b_i$ (small partial derivative in our interval of interest). To do so, we prove some properties of the function $f(a,b_i)$.
	
	\[ \frac{\partial f(a,b_i)}{\partial b_i} = \frac{a}{(2b_i+a)^2} \]
	
	\[ \frac{\partial f(a,b_i)}{\partial a} = \frac{-b_i}{(2a_i+b)^2} \]
	
	We know that $h_{t_i} > h_v$ for all indices $i$. Additionally, $h_{t_i} \ge \tau \frac{\log n}{\sqrt{n}}$ where $\tau$ is a large constant. We argue that it is not possible to independently bound each quantity $A_i$ and $B_i$ - to see why this is the case - observe that the partial derivative with respect to $b_i$ can only be upper bounded by a constant while the partial derivative with respect to $a$ can be as small as $\theta(\frac{\log n}{\sqrt{n}})$. These bounds do not suffice to prove the desired comparison between the quantities $A_i$ and $B_i$. Therefore, it is important that we carefully compare these partial derivatives when establishing that $A_i$ is significantly larger than $b_i$. More formally :  Using the intermediate value theorem on the (continuous) univariate function $f(a,h_{t_i})$ (Second argument fixed to be $h_{t_i}$) in the interval $[h_v,h'_v]$, we know there exists $x \in [h_v,h'_v]$ such that:
	
	\begin{align*}
		A_i = |f(h_v,h_{t_i}) - f(\h_v,h_{t_i})| &=  \left | \frac{\partial f(a,h_{t_i})}{\partial a} \right|_{a=x}. |h_v -\h_v| \\
		&= \left | \frac{h_{t_i}}{(2h_{t_i} + x)^2} \right|. |h_v -\h_v|
	\end{align*}

	Similarly, using the intermediate value theorem on the univariate function $f(\h_v,b)$ (first argument fixed to be $\h_{v}$) in the interval $[h_{t_i},\h_{t_i}]$, we know there exists $y \in [h_{t_i},\h_{t_i}]$ such that:
	
	\begin{align*}
		B_i = |f(\h_v,h_{t_i}) - f(\h_v,\h_{t_i})| &=  \left | \frac{\partial f(\h_v,b)}{\partial b} \right|_{b=y}. |h_{t_i}-\h_{t_i}| \\
		&= \left | \frac{\h_v}{(2y + \h_v)^2} \right|. |h_{t_i} -\h_{t_i}|
	\end{align*}
	
	Since $|\h_{t_i} - h_{t_i} | \le \theta (\fs)$ and $h_{t_i} \ge \tau \frac{\log n}{\sqrt{n}}$, we know that $2\h_{t_i} \ge h_{t_i} \ge \h_{t_i}/2 $. We also know that $\h_v \le \h_{t_i}$ by our assumption and hence $\h_v \le 2h_{t_i}$ which would imply that $x \le 2h_{t_i}$. Thus, we have $\frac{h_{t_i}}{(2h_{t_i}+x)^2} \ge \frac{h_{t_i}}{16h^2_{t_i}} = \frac{1}{16h_{t_i}}$ and $\frac{\h_v}{(2y + \h_v)^2} \le \frac{2}{h_{t_i}}$. Note that using the triangle inequality $|A_i| - |B_i| \ge ||A_i| - |B_i||$. Thus, as long as $\tau$ is large enough, we conclude that $|A_i| - |B_i| \ge \tau/20 \frac{\log n}{\sqrt{n}}$. Substituting this into inequality~*, we get a contradiction, leading us to conclude that $|h'_v - h_v| \le \Delta \frac{\log n}{\sqrt{n}}$.
	
	The final part of the proof is for the case where there exists some $i$ such that $\h_v > \h_{t_i} $. In this event, recall that our algorithm uses the smallest such height $U = \min_i \{\h_{t_i}\}$ as the final estimate $\h_v^f$ for $h_v$. We show that $|U - h_v| \le \Delta \frac{\log n}{\sqrt{n}}$ with high probability.  First, we note that $U > h_v$.Assume that  $|U - h_v| > \Delta \frac{\log n}{\sqrt{n}}$, for sake of contradiction. We will show similar bounds on the expressions $A_i$ and $B_i$ from the above proof. First, we lower bound $|A_i|$. Consider the (continuous) univariate function $f(a,h_{t_i})$ on the intervals $[h_v,U]$ and $[U,\h_v]$. Using the intermediate value theorem on both intervals, there exists $x_1 \in [h_v,U] $ and $x_2 \in [U,\h_v]$ such that:

	\begin{align*}
		A_i := |f(h_v,h_{t_i}) - f(\h_v,h_{t_i})| &=  \left | \frac{\partial f(a,h_{t_i})}{\partial a} \right|_{a=x_1}. |U -h_v| +  \left | \frac{\partial f(a,h_{t_i})}{\partial a} \right|_{a=x_2}. |\h_v -U|\\
		&\ge \left | \frac{\partial f(a,h_{t_i})}{\partial a} \right|_{a=x_1}. |U -h_v| \\
		&= \left | \frac{h_{t_i}}{(2h_{t_i} + x_1)^2} \right|. |U -h_v|
	\end{align*}
	
	Thus, with similar reasoning as before, we have $\frac{h_{t_i}}{(2h_{t_i}+x_1)^2} \ge \frac{h_{t_i}}{16h^2_{t_i}} = \frac{1}{16h_{t_i}}$ since $x_1 \le \h_{t_i}$ for all $i$.
	
	Now, to upper bound $B_i$, we consider the (continuous) univariate function $f(\h_v,b)$ on the interval $[h_{t_i},\h_{t_i}]$, we know there exists $y \in [h_{t_i},\h_{t_i}]$ such that:
	
	\begin{align*}
		B_i := |f(\h_v,h_{t_i}) - f(\h_v,\h_{t_i})| &=  \left | \frac{\partial f(\h_v,b)}{\partial b} \right|_{b=y}. |h_{t_i}-\h_{t_i}| \\
		&= \left | \frac{\h_v}{(2y + \h_v)^2} \right|. |h_{t_i} -\h_{t_i}|
	\end{align*}
	
	Consider the univariate function $g(x) = \frac{x}{(2c+x)^2}$. Note that $g(x) > 0$ if $c,x > 0$. Thus, $|g(x)|$ (for $c,x > 0$) is maximized at the same point as $g(x)$ which is at $x = 2c$. Thus, $\frac{\h_v}{(2y + \h_v)^2} \le \frac{2y}{16y^2} = \frac{1}{8y} \le \frac{1}{4h_{t_i}}$. Similar analysis as before leads us to a contradiction ( we can show that $|A_i| \ge 2|B_i|$, which suffices to prove the result).

\end{proofof}

\section{Proof of Theorem~\ref{theorem:necessary} (Lower Bound on Edge Weights)}
\label{app:lower}
	Let us assign an arbitrary ordering on the set of all possible triples of leaves $(r,s,t)$  from $1$ to $k$ where $k = \binom{n}{3}$. Let $\bQ_j^{(1)}$ (resp.
	$\bQ_j^{(2)}$) denote the random variable corresponding to the response on the $j^{th}$ query (i.e., triple) from tree $T_1$ (resp. $T_2$). Let $\bQ^{(1)}$ (resp. $\bQ^{(2)}$) denote the $\binom{n}{3}$-dimensional random variable whose $j^{th}$ coordinate is $\bQ_j^{(1)}$ (resp. $\bQ_j^{(2)}$). Both $\bQ^{(1)}$ and $\bQ^{(2)}$ follow  product distributions. To prove our theorem, it suffices to prove the following 
	\[
	\Delta(\bQ^{(1)},\bQ^{(2)}) \le 0.01. 
	\]
	Here $\Delta (\cdot, \cdot)$ refers to the total variation distance.
	The above inequality follows from Lemma~\ref{lem:total-variation} which we prove next. 
\begin{lemma}~\label{lem:total-variation}
	For $\bQ^{(1)}$ and $\bQ^{(2)}$ as described above, $\Delta(\bQ^{(1)},\bQ^{(2)}) \le 0.01$. 
\end{lemma}
\begin{proof}
	We start by partitioning the set $[k]$ into five different sets defined below: 
	\begin{enumerate}
		\item Let $L ' = L \setminus \{a,b,c\}$. Define $\mathcal{A}_1 \subseteq [k]$ to be the set of those indices which correspond to triples $(r,s,t)$ such that 
		$|\{r,s,t\} \cap L'| \ge 2$. 
		\item Define $\mathcal{A}_2 \subseteq [k]$ to be the set of those indices which correspond to triples $(r,b,c)$ where $r \in L'$. 
		\item Let $j^\ast$ be the index which corresponds to the triple $(a,b,c)$. Let $\mathcal{A}_3 = \{j^\ast\}$. 
		\item Let $\mathcal{A}_4$ be the set of those indices which correspond to the triples of the form $(a,b,x)$ for $x \in L'$. 
		\item Let $\mathcal{A}_5$ be the set of those indices which correspond to the triples of the form $(a,c,x)$ for $x \in L'$.
	\end{enumerate}
	In order to analyze $\Delta(\bQ^{(1)},\bQ^{(2)})$, it is more convenient to look at the notion of Hellinger distance (see Definition~\ref{def:Hellinger}). 
	We start with the following easy claims. 
	\begin{claim}~\label{clm:identical-variables}
		Let $j \in \mathcal{A}_1$. Then the random variables $\bQ_j^{(1)}$ and  $\bQ_j^{(2)}$ are identical. Thus, 
		$H^2(\bQ_j^{(1)}, \bQ_j^{(2)}) =0$. 
	\end{claim}
	\begin{proof}
		Let $j$ correspond to the triple $(r,s,t)$. 
		Note that the random variables $\bQ_j^{(1)}$ and  $\bQ_j^{(2)}$ are just dependent on the pairwise distances between the leaves $r$, $s$ and $t$. It is easy to observe that as long as two of these leaves are in $L'$, their pairwise distances are same both in trees $T_1$ and $T_2$.  This proves the claim. 
	\end{proof}
	\begin{claim}~\label{clm:identical-variables-prime}
		Let $j \in \mathcal{A}_2$. Then the random variables $\bQ_j^{(1)}$ and  $\bQ_j^{(2)}$ are identical. Thus, 
		$H^2(\bQ_j^{(1)}, \bQ_j^{(2)}) =0$. 
	\end{claim}
	\begin{proof}
		Let $j$ correspond to the triple $(r,c,b)$. 
		As in Claim~\ref{clm:identical-variables}, the pairwise distances between the vertices $r$, $c$ and $b$ are the same in both $T_1$ and $T_2$. The claim now follows. 
	\end{proof} 
	Using Claim~\ref{clm:identical-variables-prime} and Claim~\ref{clm:identical-variables}, it follows tha
\end{proof}
\begin{claim}~\label{clm:diff-probab1}
	Let $\mathcal{A}_3 = \{j^\ast\}$. Then, $H^2(\bQ_{j^\ast}^{(1)}, \bQ_{j^\ast}^{(2)}) \le \frac{4\rho^2}{n}$.
\end{claim}
\begin{proof}
	Observe that both random variables $\bQ_{j^\ast}^{(1)}$ and  $\bQ_{j^\ast}^{(2)}$ are supported on the set $\{(a,b), (b,c) , (a,c) \}$. Let $p_1 := \Pr[\bQ_{j^\ast}^{(1)} = (a,b)]$, $p_2 := \Pr[\bQ_{j^\ast}^{(1)} = (b,c)]$ and $p_3 := \Pr[\bQ_{j^\ast}^{(1)} = (a,c)]$. Similarly, let $q_1 := \Pr[\bQ_{j^\ast}^{(2)} = (a,b)]$, $q_2 := \Pr[\bQ_{j^\ast}^{(2)} = (b,c)]$ and $q_3 := \Pr[\bQ_{j^\ast}^{(2)} = (a,c)]$.

	Let $d_1(\cdot, \cdot)$ be the distance metric on tree $T_1$ and $d_2(\cdot, \cdot)$ be the distance metric on tree $T_2$.  Now, define $\alpha := d_1(a,b) = d_2(a,c)$ and $\beta: = d_1(a,c) = d_1(b,c) = d_2(a,b) =d_2(b,c)$.

	Observe that $\alpha \ge \frac{2}{3}$ since $d_1(a,p) = d_2(a,p) =\frac{1}{3}$. Also, $\beta - \alpha = \frac{2\rho}{\sqrt{n}}$. 
	
	\[H^2(\bQ_{j^\ast}^{(1)},\bQ_{j^\ast}^{(2)}) = \frac{1}{2} ( (\sqrt{p_1} - \sqrt{q_1} )^2 + (\sqrt{p_2} - \sqrt{q_2} )^2 + (\sqrt{p_3} - \sqrt{q_3} )^2 )\].
	
	We rewrite the probabilities in terms of the distances $\alpha$ and $\beta$:
	\[
	p_1 = q_3 = \frac{2 \beta}{2(\alpha + 2 \beta)}; \ 
	p_2 = p_3 = q_1 = q_2 = \frac{\alpha + \beta}{2(\alpha + 2 \beta)}
	\]

	Next, we upper bound the quantity: $\sqrt{p_1} - \sqrt{q_1}$.
	
	\begin{align*}
		\sqrt{p_1} - \sqrt{q_1} &=  \sqrt{\frac{2 \beta}{2(\alpha + 2 \beta)}} - \sqrt{\frac{\alpha + \beta}{2(\alpha + 2 \beta)}} \\
		&= \frac{\sqrt{2\beta} - \sqrt{\alpha+\beta}}{\sqrt{2(\alpha+2\beta)}} \\
		&\le \sqrt{\frac{4}{3}} (\sqrt{2\beta} - \sqrt{\alpha+\beta}) \qquad \text{( since $\alpha \ge \frac{2}{3}$)} \\
		&= \sqrt{\frac{4}{3}} \frac{\beta - \alpha}{(\sqrt{2\beta} + \sqrt{\alpha+\beta})} \qquad \text{( multiplying and dividing by $(\sqrt{2\beta} + \sqrt{\alpha+\beta})$)} \\
		&\le \frac{2\rho}{\sqrt{n}}  \qquad \text{( since $\alpha \ge \frac{2}{3}$)}
	\end{align*}
	
	We know $p_2 = q_2$. The final term of interest is $|\sqrt{p_3} - \sqrt{q_3}|$. Since $p_1 = q_3$ and $p_3 = q_1$, we get the upper bound - 
	$|\sqrt{p_3} - \sqrt{q_3}| \le \frac{2\rho}{\sqrt{n}}$.
	
	Putting together these inequalities, we get 
	$H^2(\bQ_{j^\ast}^{(1)},\bQ_{j^\ast}^{(2)}) \le \frac{1}{2} (\frac{8 \rho^2 }{n})$.
\end{proof}






\begin{claim}
	~\label{clm:diff-probab2}
	Let $j \in \mathcal{A}_4$. Then, $H^2(\bQ_{j}^{(1)}, \bQ_{j}^{(2)}) \le \frac{\rho^2}{4n}$.
\end{claim}
\begin{proof}
	Let $j$ correspond to the triple $(a,b,x)$.  The support space of both $\bQ_{j}^{(1)}$ and $\bQ_{j}^{(2)}$ is $\{(a,b), (b,x) , (a,x) \}$. Let $p_1 := \Pr[\bQ_{j}^{(1)} = (a,b)]$, $p_2 := \Pr[\bQ_{j}^{(1)} = (b,x)]$ and $p_3 := \Pr[\bQ_{j}^{(1)}= (a,x)]$. Similarly, let $q_1 := \Pr[\bQ_{j}^{(2)} = (a,b)]$, $q_2 := \Pr[\bQ_{j}^{(2)} = (b,x)]$ and $q_3 := \Pr[\bQ_{j}^{(2)} = (a,x)]$.

	As in Claim~\ref{clm:diff-probab1}, let $d_1(\cdot, \cdot)$ be the distance metric on tree $T_1$ and $d_2(\cdot, \cdot)$ be the distance metric on $T_2$. 
	Let $ \alpha :=d_1(a,b) $ and $ \beta := d_2(a,b)$. We have $d_1(a,x) = d_1(b,x) = d_2(a,x) = d_2(b,x) = 2$, since the associated unique path goes through the root for each of these pairs. 
	Finally,  $\beta - \alpha = \frac{2\rho}{\sqrt{n}}$. Using the definition of Hellinger distance, we have : 
	
	\[H^2(\bQ_{j}^{(1)},\bQ_{j}^{(2)}) = \frac{1}{2} ( (\sqrt{p_1} - \sqrt{q_1} )^2 + (\sqrt{p_2} - \sqrt{q_2} )^2 + (\sqrt{p_3} - \sqrt{q_3} )^2 )\].
	Thus, we have, 
	\[
	p_1  = \frac{4}{2(4 + \alpha)}; \ q_1  = \frac{4}{2(4+\beta)}; \ p_2 = p_3  = \frac{2 + \alpha }{2(4 + \alpha)}; \ q_2 = q_3  = \frac{2 + \beta }{2(4 + \beta)}
	\]
	Next, we upper bound the quantity  $\sqrt{p_1} - \sqrt{q_1}$ as follows: 
	\begin{align}
		\sqrt{p_1} - \sqrt{q_1} &= \sqrt{2} \left( \frac{\sqrt{4+\beta} - \sqrt{4+\alpha}}{\sqrt{(4+\alpha)(4+\beta)}} \right) \nonumber\\
		&\le \frac{1}{2 \sqrt{2}} (\sqrt{4+\beta} - \sqrt{4+\alpha})  \nonumber\\
		&\le \frac{1}{2 \sqrt{2} (\sqrt{4+\beta} + \sqrt{4+\alpha})  } (\beta - \alpha)  \nonumber\\
		&\le \frac{1}{4 \sqrt{2}} \frac{ \rho}{\sqrt{n}} \label{eq:abc1}
	\end{align}
	Next, we upper bound $\sqrt{q_2} - \sqrt{p_2}$ as follows:
	\begin{align}
		\sqrt{q_2} - \sqrt{p_2} &= \frac{1}{\sqrt{2}} \frac{\sqrt{(2+\beta)(4+\alpha)}- \sqrt{(4+\beta)(2+\alpha)}}{\sqrt{(4+\alpha)(4+\beta)}}  \nonumber\\
		&\le \frac{1}{4 \sqrt{2}} (\sqrt{8 + \alpha \beta + 4 \beta + 2 \alpha} - \sqrt{8 + \alpha \beta + 2 \beta + 4 \alpha} )  \nonumber\\
		&= \frac{1}{ 4\sqrt{2} (\sqrt{8 + \alpha \beta + 4 \beta + 2 \alpha} + \sqrt{8 + \alpha \beta + 2 \beta + 4 \alpha} ) } (2\beta - 2\alpha)  \nonumber\\
		&\le \frac{1}{8 \sqrt{2}} \frac{\rho}{\sqrt{n}} \label{eq:abc2}
	\end{align}
	By an identical analysis, we have $\sqrt{q_3} - \sqrt{p_3} \le \frac{1}{8 \sqrt{2}} \frac{\rho}{\sqrt{n}}$. Combining this with \eqref{eq:abc1} and \eqref{eq:abc2} gives us the claim.
\end{proof}

Using an analysis identical to Claim~\ref{clm:diff-probab2}, we also get the following Claim. 
\begin{claim}~\label{clm:diff-probab3}
	Let $j \in \mathcal{A}_5$. Then, $H^2(\bQ_{j}^{(1)}, \bQ_{j}^{(2)}) \le \frac{\rho^2}{4n}$.
\end{claim}

Putting together these claims, we can use Lemma~\ref{lemma:producthell} to show that 
$H(\bQ^{(1)}, \bQ^{(2)}) \le \frac{0.01}{\sqrt{2}}$, and then use Lemma~\ref{lemma:comparisontvd} to conclude that $\Delta(\bQ^{(1)},\bQ^{(2)}) \le 0.01$.

\section{Concentration Bounds and Measures of Statistical Distance}
\label{sec:con}

We list here some standard concentration results, that will be used to aggregate the results of stochastically independent queries

\begin{theorem}[Generalized Hoeffding Bound]
\label{theorem:generalized_hoeffding}
Let $\by_1, \by_2 ,.., \by_k$ be $k$ independent random variables, with each variable $\by_i$ having range $[a_i,b_i]$ and mean $y_i$  then: 
\[  \Pr \big[ \big|\sum_{i=1}^k y_i- \sum_{i=1}^k \by_i \big| \le t \big] \ \geq \  1 - e^{\frac{2t^2}{\sum_i (b_i -a_i)^2}}\] 
\end{theorem}

We state below a special case of the generalized Hoeffding bound that we use repeatedly in proofs, referring to it as the Hoeffding bound.
\begin{theorem}[Hoeffding Bound]
\label{theorem:hoeffding}
Let $\by_1, \by_2 ,.., \by_k$ be $k$ iid random variables between $0$ and $1$, each with mean $y$, then: 
\[  \Pr \big[ \big|y- \frac{\sum_{i=1}^k \by_i}{k} \big| \le 4 \sqrt{\frac{\log n}{k}} \big] \ \geq \  1 - \frac{1}{n^6}\] 
\end{theorem}


We also list some measures of statistical distance and connections between them, that will be employed in our lower bounds proofs. We formally define total variation distance and Hellinger distance for discrete distributions.

\begin{definition}[Total Variation Distance]
\label{def:tvd}
Let $X$ and $Y$ be discrete distributions, having weight $p_{x_i}$ and $p_{y_j}$ respectively on points $z_1$, $z_2\cdots $. Then, the total variation distance between $X$ and $Y$, denoted by $\Delta(X,Y)$ is defined as:
\[\Delta(X,Y) := \frac{1}{2}\sum_{i=1}^\infty |p_{x_i}-p_{y_i}| \]
 An alternative, equivalent definition is as follows: Let the sample space of the two distributions $X,Y$ be $\Omega$, then:
 
 $\Delta(X,Y) = \sup_{A \in \Omega} |X(A) - Y(A)|$.
\end{definition}

\begin{definition}[Hellinger Distance]~\label{def:Hellinger}
Let $X$ and $Y$ be discrete distributions, having weight $p_{x_1}, p_{x_2}, \cdots$ and $p_{y_1}, p_{y_2}, \cdots$ respectively on points $z_1$, $z_2\cdots $. Then, the total variation distance between $X$ and $Y$, denoted by $H(X,Y)$ is defined as:
\[H(X,Y) := \frac{1}{\sqrt{2}}\sqrt{\sum_{i=1}^\infty (\sqrt{p_{x_i}}-\sqrt{p_{y_i}})^2} \]

\end{definition}

We use the following two useful lemmas about Total Variation Distance Hellinger Distance from Barak et. al.~\cite{barak2008rounding}.

\begin{lemma}[\cite{pollard_2001}]
\label{lemma:comparisontvd}
For two distributions $X$ and $Y$:
\[H^2(X,Y) \le \Delta(X,Y) \le \sqrt{2} H(X,Y) \]
\end{lemma}

\begin{lemma}[\cite{barak2008rounding}]
\label{lemma:producthell}
Let $X_1,X_2\cdots X_n$ and $Y_1,Y_2\cdots Y_n$ be two families of distributions. Then,
\[H(X_1 \oplus X_2 \oplus \cdots X_n, Y_1 \oplus Y_2 \oplus \cdots Y_n) \le \sum_{i=1}^n H^2(X_i,Y_i) \]

where $X \oplus Y$ denotes the product of two distributions $X$ and $Y$, generated by taking independent samples of $X$ and $Y$.

\end{lemma}


 \section{Figures}
 \label{sec:figures}

 \begin{figure}
    \centering
    \includegraphics[width=6cm]{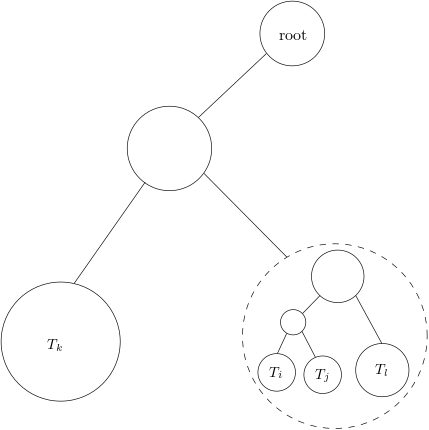}
    \caption{Sibling-Tree Pair between Non Sibling-Tree Pair}
    \label{fig:score_winner}
\end{figure}

\begin{figure}
	\centering
	\begin{minipage}{0.5\textwidth}
		\centering
		\includegraphics[width=.8\linewidth]{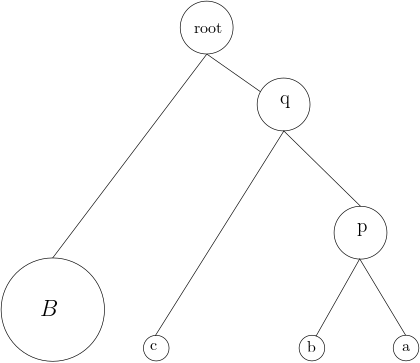} 
		\captionof{figure}{Lower Bound Instance Tree $T_1$}
		\label{fig:Lower Bound Instance Tree 1}
	\end{minipage}%
	\begin{minipage}{.5\textwidth}
		\centering
		\includegraphics[width=.8\linewidth]{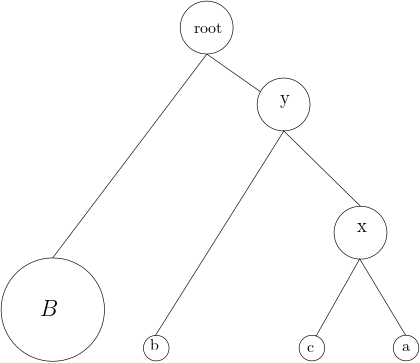}
		\captionof{figure}{Lower Bound Instance Tree 2}
		\label{fig:Lower Bound Instance Tree $T_2$}
	\end{minipage}
\end{figure}

\end{document}